\documentclass[final,twoside,11pt]{entics} 
\usepackage{enticsmacro}
\usepackage{amsmath}
\usepackage{amsfonts}
\usepackage{proof}
\usepackage{cmll}
\usepackage{stmaryrd}
\usepackage[shortlabels]{enumitem}
\usepackage{graphicx}
\usepackage{hyperref}
\newtheorem{notation}{Notation}

\newcommand{\imp}{\quad\Rightarrow\quad}
\newcommand{\comb}[1]{\mathbf{#1}}
\newcommand{\Om}{\boldsymbol{\Omega}}
\newcommand{\bsub}{\begin{enumerate}}
\newcommand{\esub}{\end{enumerate}}
\renewcommand{\bar}{\begin{array}}
\newcommand{\ear}{\end{array}}
\newcommand{\nat}{{\mathbb N}}
\newcommand{\st}{\mid}
\newcommand\set[1]{\{#1\}}
\newcommand{\lam}{\ensuremath{\lambda}} 
\newcommand{\tuple}[1]{\langle #1 \rangle}
\newcommand{\Tuple}[1]{\left\langle\begin{aligned} #1 \end{aligned}\right\rangle}
\newcommand{\cI}{\mathcal{I}}
\newcommand{\cR}{\mathcal{R}}

\newcommand{\und}{\ \land \ }
\newcommand{\E}{\mathsf{E}}
\newcommand{\Val}{\mathrm{Val}}
\newcommand{\redd}{\Downarrow}
\newcommand{\Var}{\mathrm{Var}}
\newcommand{\PCF}{{\sf PCF}}
\newcommand{\tint}{\mathsf{int}}
\newcommand{\pred}{\mathbf{pred}\,}
\renewcommand{\succ}{\mathbf{succ}\,}
\newcommand{\ifterm}[3]{\mathbf{ifz}(#1,#2,#3)}
\newcommand{\fix}{\mathbf{fix}\,}
\newcommand{\num}{\underline}
\newcommand{\subst}[2]{[#2/#1]}
\newcommand{\FV}[1]{\mathrm{FV}(#1)}

\newcommand{\esubst}[2]{#1 \langle #2 \rangle}
\newcommand{\substseq}[2]{#1 \triangleright #2}
\newcommand{\EPCF}{{\sf EPCF}}

\newcommand{\dom}{\mathrm{dom}}

\newcommand{\valrule}{(\mathrm{nat})}
\newcommand{\PCFvalrule}{(\mathrm{val})}
\newcommand{\funrule}{(\mathrm{fun})}
\newcommand{\varrule}{(\mathrm{var})}
\newcommand{\predrule}{(\mathrm{pr})}
\newcommand{\succrule}{(\mathrm{sc})}
\newcommand{\predzrule}{(\mathrm{pr_0})}
\newcommand{\ifzzrule}{(\mathrm{ifz_0})}
\newcommand{\ifzrule}{(\mathrm{ifz_{>0}})}
\newcommand{\betarule}{(\beta_v)}
\newcommand{\fixrule}{(\mathrm{fix})}

\newcommand{\Types}{\mathbb{T}}
\newcommand{\size}[1]{\vert#1\vert}

\newcommand{\Addrs}{\mathbb{A}}
\newcommand{\Tapes}[1][\Addrs]{\mathcal{T}_{#1}}
\newcommand\cM[1][\Addrs]{\mathcal{M}_{#1}}
\newcommand{\appT}[2]{#1\,@\,#2\,}
\newcommand{\append}[2]{\appT{#1}{{[}#2{]}}}
\newcommand{\Lookinv}[1]{\#^{-1}(#1)}
\newcommand{\Lookup}{\#}
\newcommand{\App}[2]{#1\cdot #2}
\newcommand{\Cons}[2]{#1::#2}
\newcommand{\nnat}[1][+]{\nat^{#1}}
\newcommand{\repl}[2]{[#1:=#2]}
\newcommand{\stuck}[1]{\mathsf{stuck}(#1)}
\newcommand{\trans}[3]{\llbracket #1\rrbracket^{#3}}
\newcommand{\etrans}[3]{\llparenthesis #1 \rrparenthesis^{#3}}
\newcommand{\mPredn}[1]{\mach{Pred}_{#1}}
\newcommand{\mSuccn}[1]{\mach{Succ}_{#1}}
\newcommand{\mIfZn}[1]{\mach{Ifz}_{#1}}
\newcommand{\mProj}[2]{\mach{Pr}_{#2}^{#1}}
\newcommand{\mAppn}[1]{\mach{Apply}_{#1}}
\newcommand{\mYn}[1]{\mach{Y}_{#1}}

\newcommand{\mach}{\mathsf} 
\newcommand{\Null}{\varnothing}
\newcommand{\val}[1]{\oc #1}
\newcommand{\mM}{\mach{M}}
\newcommand{\mN}{\mach{N}}

\newcommand{\ins}[1]{\mathtt{#1}}
 \newcommand{\Call}[1]{\ins{Call}~#1}
\newcommand{\RaS}[1]{\ins{Load}~#1}
\newcommand{\Load}[1]{\ins{Load}~#1}
\newcommand{\Ifz}[4]{#4\shortleftarrow \ins{Test}(#1,\,#2,\,#3)}
\newcommand{\Pred}[2]{#2\shortleftarrow \ins{Pred}(#1)}
\newcommand{\Succ}[2]{#2\shortleftarrow \ins{Succ}(#1)}
 \newcommand{\Apply}[3]{#3\shortleftarrow \ins{App}(#1,\,#2)}

\newcommand{\redh}{\to_{\mach{c}}}
\newcommand{\reddh}{\twoheadrightarrow_{\mach{c}}}
\newcommand{\convh}{\leftrightarrow_\mach{c}}
\newcommand{\err}{\texttt{err}}
\newcommand{\reddd}{\redd_d}

 \renewcommand{\size}[1]{\vert{#1}\vert}

\def\conf{MFPS 2022} 	
\volume{1}			
\def\volu{1}			
\def\papno{8}			
\def\mydoi{{\normalshape  \href{https://doi.org/10.46298/entics.10533}{10.46298/entics.10533}}}
\def\copynames{B. Intrigila, G. Manzonetto, N. M\"unnich} 
\def\runtitle{Exxtended Addresing Machines for PCF,...}
\def\lastname{Intrigila, Manzonetto, M\"unnich}
\begin{document}
\begin{frontmatter}
  \title{Extended Addressing Machines for PCF,\\ with Explicit Substitutions}  

  \author{Benedetto Intrigila\thanksref{intremail}}
  \address{Dipartimento di Ingegneria dell'Impresa,\\ University of Rome ``Tor Vergata'', Italy} 
  \thanks[intremail]{Email: 
  \href{mailto:benedetto.intrigila@uniroma2.it}{\texttt{\normalshape
        benedetto.intrigila@uniroma2.it}}}
  
\author{Giulio Manzonetto\thanksref{PPS}\thanksref{myemail}}
  \author{Nicolas M\"unnich\thanksref{Flavien}\thanksref{coemail}}
  \address{Univ. USPN, Sorbonne Paris Cit\'e,\\ LIPN, UMR 7030, CNRS, F-93430 Villetaneuse, France.}
  \thanks[PPS]{Partly supported by ANR Project PPS, ANR-19-CE48-0014.} 
  \thanks[Flavien]{Partly supported by ANR JCJC Project CoGITARe, ANR-18-CE25-0001.} 
  \thanks[myemail]{Email:
    \href{mailto:giulio.manzonetto@lipn.univ-paris13.fr} {\texttt{\normalshape
        giulio.manzonetto@lipn.univ-paris13.fr}}} \thanks[coemail]{Email:
    \href{mailto:munnich@lipn.univ-paris13.fr} {\texttt{\normalshape
        munnich@lipn.univ-paris13.fr}}}
\begin{abstract} 
 Addressing machines have been introduced as a formalism to construct models of the pure, untyped \lam-calculus.
We extend the syntax of their programs by adding instructions for executing arithmetic ope\-ra\-tions on natural numbers, and introduce a reflection principle allowing certain machines to access their own address and perform recursive calls.
We prove that the resulting extended addressing machines naturally model a weak call-by-name \PCF{} with explicit substitutions.
Finally, we show that they are also well-suited for representing regular \PCF{} programs (closed terms) computing natural numbers.
\end{abstract}
\end{frontmatter}
\renewcommand{\vec}{\pol}
\section*{Introduction}

Turing machines (TM) and \lam-calculus constitute two fundamental formalisms in theoretical computer science.
Because of the difficulty in emulating higher-order calculations on a TM, their equivalence on partial numeric functions is not obtained directly, but rather composing different encodings. 
As a consequence, no model of \lam-calculus (\emph{\lam-model}) based on TM's has arisen in the literature so far.
Recently, Della Penna et al.\ have successfully built a \lam-model based on so-called \emph{addressing machines} (AM) \cite{DellaPennaIM21}.
The intent is to propose a model of computation, alternative to von Neumann architecture, where computation is based on communication between machines rather than performing local operations.
In fact, these machines are solely capable of manipulating the addresses of other machines---this opens the way for modelling higher-order computations since functions can be passed via their addresses. 
An AM can read an address from its input-tape, store in a register the result of applying an address to another and, finally, pass the execution to another machine by calling its address (possibly extending its input-tape).
The set of instructions is deliberately small, to identify the minimal setting needed to represent \lam-terms.
The downside is that performing calculations on natural numbers is as awkward as using Church numerals in \lam-calculus.

{\bfseries Contents.} In this paper we extend the formalism of AM's with a set of instructions representing basic arithmetic operations and conditional tests on natural numbers.
As we are entering a world of machines and addresses, we need specific machines to represent numerals and assign them recognizable addresses.
Finally, in order to model recursion, we rely on the existence of machines representing {\em fixed point combinators}. 
These machines can be programmed in the original formalism but we can avoid any dependency on {\em self-application} by manipulating the addressing mechanism so that they have access to {\em their own address}. This can be seen as a very basic version of the reflection principle which is present in some programming languages. 
We call the resulting formalism \emph{extended addressing machines} (EAMs).

Considering these features, one might expect EAMs to be well-suited for simu\-la\-ting Plotkin's {\em Programming Computable Functions} (\PCF)~\cite{Plotkin77}, a simply typed \lam-calculus with constants, arithmetical operations, conditional testing and a fixed point combinator.
A \PCF{} term of the form $(\lam x.M)N$ can indeed be translated into a machine $\mach{M}$ reading as input ($x$) from its tape the address of~$\mach{N}$. 
As $\mach{M}$ has control over the computation, it naturally models a weak leftmost call-by-name evaluation.
However, while in the contractum $M\subst{x}{N}$ of the redex the substitution  is instantaneous, $\mach{M}$ needs to pass the address of $\mach{N}$ to the machines representing the subterms of $M$, with the substitution only being performed if $\mach{N}$ gains control of the computation.
As a result, rather than \PCF{}, EAMs naturally emulate the behavior of \EPCF---a weak call-by-name \PCF{} with explicit substitutions that are only performed ``on demand'', as in \cite{LevyM99}.
We endow EAMs with a typing mechanism based on simple types and define a type-preserving translation from well-typed \EPCF{} terms to EAMs.
Subsequently, we prove that also the operational behavior of \EPCF{} is faithfully represented by the translation.
Finally, by showing the equivalence between \PCF{} and \EPCF{} on terminating programs of type $\tint$, we are capable of drawing conclusions for the original language \PCF.

In this paper we mainly focus on the properties of the translation, but our long-term goal is to construct a sequential model of higher-order computations. The problem of finding a fully abstract model of \PCF{} was originally proposed by Robin Milner in \cite{Milner77} and is a difficult one.
A model is called \emph{fully abstract} (FA) whenever two programs sharing the type $\alpha$ get the same denotation in the model if and only if they are observationally indistinguishable when plugged in the same context $C[]$ of type $\alpha\to\tint$. 
Therefore, a FA model provides a semantic characterization of the observational equivalence of \PCF.
Quoting from~\cite{AbramskyJM00}:
\begin{quote} ``the problem is to understand what should be meant by a semantic characterization [\ldots] Our view is that the essential content of the problem, what makes it important, is that it calls for a semantic characterization of sequential, functional computation at higher-types''.
\end{quote}
A celebrated result is that FA models of \PCF{} can be obtained by defining suitable categories of games \cite{AbramskyMJ94,AbramskyJM00,HylandO00}.
Preliminary investigations show that EAMs open the way to construct a more `computational' FA model. 
E.g., in \cite{Milner77}, the model construction starts with first-order definable functions and requires---to cope with fixed point operator---the addition of extra `limit points' to ensure that the resulting partial order is direct complete. 
In the game semantics approach the fixed point operator is treated similarly, namely via its canonical interpretation in a cpo-enriched Cartesian closed category~\cite{AbramskyMJ94}.
 On the contrary, in our approach no limit construction is required to give the fixed point operator a meaning. 
 The fact that EAMs possess a given recursor having its own address stored inside is easily obtained from a mathematical point of view and, as argued above, can be seen as an abstract view of the usual implementation of recursion. 
 We believe this new point of view may increase our understanding of \PCF{} observational equivalence.

{\bfseries Outline.} The paper is organized as follows. In Section~\ref{sec:EPCF} we introduce the language \EPCF{} along with its syntax, simply typed assignment system and associated (call-by-name) big-step operational semantics. 
In Section~\ref{sec:EAMS} we define EAMs (no familiarity with \cite{DellaPennaIM21} is assumed) and introduce their operational semantics. In Section~\ref{sec:Types} we describe a type--checking algorithm for determining whether an EAM is well-typed.
In Section~\ref{sec:trans} we present our main results, namely: 
\bsub[(i)]\item the translation of a well-typed \EPCF{} term is an EAM typable with the same type (Theorem~\ref{thm:typabilitytransfers}); 
\item if an \EPCF{} term reduces to a value, then their translations as machines are interconvertible (Thm.~\ref{thm:equivalence}); 
\item the operational semantics of \PCF{} and \EPCF{} coincide on terminating programs of type $\tint$ (Thm.~\ref{thm:delayed_equivalence}); 
\item  the translation of a \PCF{} program computing a number is an EAM evaluating the corresponding numeral (Theorem~\ref{thm:main}).
\esub

{\bf Related works.} A preliminary version of AMs was introduced in Della Penna's MSc thesis \cite{DellaPennaTh} in order to model computation as communication between distinguished processes by means of their addresses. 
They were subsequently refined in~\cite{DellaPennaIM21} with the theoretical purpose of constructing a model of \lam-calculus.
Similarly, our paper should be seen as a first step towards the construction of a denotational model of~\PCF. 
Thus, the natural comparison is\footnote{The reader interested in a comparison with other abstract machines or formalisms is invited to consult \cite{IMM22}.
} with other models rather than other machine-based formalisms that have been proposed in the literature (e.g., call-by-name: SECD \cite{Landin64}, KAM \cite{Krivine07}, call-by-need: TIM~\cite{FairbairnW87}, Lazy KAM~\cite{Cregu91,Lang07}); call-by-value: ZINC~\cite{Leroy90}) from which they differ at an implementational level.

Compared with models of \PCF{} based on Scott-continuous functions~\cite{Milner77,BerryCL85,Curien07}, EAMs provide a more operational interpretation of a program and naturally avoid parallel features that would lead to the failure of FA as in the continuous semantics. 
Compared with Curien's sequential algorithms \cite{Curien92} and categories of games~\cite{AbramskyMJ94,HylandO00} they share the intensionality of programs' denotations, while presenting an original way of modelling sequential computation.
The model based on AMs also bares \emph{some} similarities with the categories of assembly used to model PCF \cite{LongleyTh}, mostly on a philosophical level, in the sense that these models are based on the `codes' (rather than addresses) of recursive functions realizing a formula ($\cong$ type). 

Concerning explicit substitutions we refer to the pioneering articles \cite{AbadiCCL90,AbadiCCL91,CurienHL96,LevyM99}.
Explicit substitutions have been barely considered in the context of \PCF---with the notable exception of \cite{SeamanI96}.

\section{Preliminaries}\label{sec:EPCF}
\begin{figure*}[t!]
\begin{minipage}{\textwidth}
\[\bar{c}
\bar{ccc}
\infer[{{\valrule}}]{\substseq{\sigma}\num n \reddd \num n }{ \num n\in \nat}
	&
	\infer[{{\funrule}}]{\substseq{\sigma}{\lam x.\esubst{M}{\rho}} \reddd \lam x.\esubst{M}{\sigma + \rho} }{}
	&
	\infer[\varrule]{\substseq{\sigma}x\reddd V}{\sigma(x) = (\rho, N) & \substseq{\rho}N\reddd V}
	\\[1ex]
	\infer[\fixrule]{\substseq \sigma {\fix M} \reddd V}{\substseq{\sigma}{ M \cdot (\fix M)} \reddd V}
	&
	\infer[\ifzzrule]{\substseq\sigma\ifterm M{N_1}{N_2} \reddd  V_1}{\substseq\sigma M\reddd \num{0} & \substseq{\sigma}{N_1}\reddd V_1 }
	&
	\infer[\ifzrule]{\substseq\sigma\ifterm M{N_1}{N_2} \reddd  V_2}{\substseq\sigma M\reddd \num{n+1} & \substseq{\sigma}{N_2}\reddd  V_2}
	\\[1ex]
	\infer[\predrule]{\substseq{\sigma}{ \pred M}  \reddd \num n}{\substseq{\sigma}{ M}  \reddd \num{n+1}}
	&
	\infer[\predzrule]{\substseq{\sigma}{ \pred M} \reddd \num {0}}{\substseq{\sigma}{ M}  \reddd \num{0}}	
	&
	\infer[\succrule]{\substseq{\sigma}{ \succ M}  \reddd  \num {n+1}}{\substseq{\sigma}{M}  \reddd \num{n}}
	\ear
	\\[10ex]
	\infer[\betarule]{\substseq\sigma{ M \cdot N} \reddd V }{\substseq\sigma M\reddd \esubst{\lam x.M'}{\rho} &\substseq{\rho + [ x \leftarrow (\sigma, N)]}{ M'}\reddd V}
\ear
\]
\end{minipage}
\caption{The big-step operational semantics of \EPCF.}\label{fig:PCFheadDelay}
\end{figure*}

The paradigmatic programming language \PCF~\cite{Plotkin77} is a simply typed \lam-calculus enriched with constants representing natural numbers, the fundamental arithmetical operations, an if-then-else conditional instruction, and a fixed-point operator. 
We give \PCF{} for granted and rather present \EPCF, an extension of \PCF{} with explicit substitutions~\cite{LevyM99}. 
We draw conclusions for the standard \PCF{} by exploiting the fact that they are equivalent on programs (closed terms) of type $\tint$.

\begin{definition}\label{def:EPCF} Consider fixed a countably infinite set $\Var$ of variables.
\EPCF{} \emph{terms} and \emph{explicit substitutions} are defined by (for $n\ge 0$ and $\vec x\in\Var$):
	\[
	\bar{lcl}
		L,M,N&::=&~\,x\mid M\cdot N \mid \lam x.\esubst{M}{\sigma}
		\mid \mathbf{0} \mid \pred M \mid \succ M\mid \ifterm L{M}{N} \mid \fix M\\
		\sigma,\rho &::=&\left[x_1\leftarrow(\sigma_1,M_1),\dots,x_n\leftarrow(\sigma_n,M_n)\right]
	\ear
	\]
\end{definition}

As is customary, $M \cdot N$ stands for the \emph{application} of a term $M$ to its argument $N$, $\mathbf{0}$ represents the natural number 0, $\mathbf{pred}$ and $\mathbf{succ}$ indicate the predecessor and successor respectively, $\mathbf{ifz}$ is the conditional test on zero, and finally, $\mathbf{fix}$ is a fixed-point operator. We assume that application -- often denoted as juxtaposition -- associates to the left and has higher precedence than abstraction. 
Concerning $\lam x.\esubst{M}{\sigma}$, it represents an \emph{abstraction} where $\sigma$ is a list of assignments from variables to \emph{closures} (terms with the associated substitutions), where each variable can only have one closure assigned to it. 

\begin{definition}
\bsub\item
In an explicit substitution \[\sigma = [x_1 \leftarrow (\sigma_1, M_1),\dots,x_n \leftarrow (\sigma_n,M_n)]\] the $x_i$'s are assumed to be fresh and distinguished. 
\item By (i), we can define $\sigma(x_i) = (\sigma_i, M_i)$.
\item
The \emph{domain} of $\sigma$ is given by $\dom(\sigma)=\set{x_1,\dots,x_n}$.
\item
We write $\sigma+\rho$ for the concatenation of $\sigma$ and $\rho$, and in this case we assume $\dom(\sigma)\cap\dom(\rho) = \emptyset$.
\esub
\end{definition}
The set $\FV{M}$ of \emph{free variables} of an \EPCF{} term $M$ is defined as usual, except for the abstraction case $\FV{\lam x.\esubst{M}{\sigma}} = \FV{M}-(\set{x}\cup\dom(\sigma))$.
The term $M$ is \emph{closed} if $\FV{M} = \emptyset$, and in that case it is called an \emph{\EPCF{} program}.

Hereafter terms are considered up to renaming of bound variables. Therefore the symbol $=$ will denote syntactic equality up to \emph{$\alpha$-conversion}.

\begin{notation}
\bsub
\item For every $n\in\nat$, we let $\num n = \succ^n(\mathbf{0})$. In particular, $\num 0$ is an alternative notation for $\mathbf{0}$.
\item 	As a syntactic sugar, we write $\lam x.M$ for $\lam x.\esubst{M}{}$. With this notation in place, \emph{\PCF~terms} are simply \EPCF{} terms containing empty explicit substitutions.

\item 	
	For $n\in\nat$, we often write $\lam x_1\dots \lam x_n.M$ as $\lam x_1\dots x_n.M$, or even $\lam\vec x.M$ when $n$ is clear from the context.
	Summing up, and recalling that $\cdot$ is left associative, $\lam x_1x_2x_3.L\cdot M\cdot N$ stands for  $\lam x_1.\esubst{(\lam x_2.\esubst{(\lam x_3.\esubst{((L\cdot M)\cdot N)}{})}{})}{}$.
\item $M\subst{x}{N}$ denotes the capture-free substitution of $N$ for all free occurrences of $x$ in $M$.
\esub
\end{notation}

\begin{example}\label{ex:PCFterms} 
We introduce some notations for the following ({\sf E})\PCF{} programs, that will be used as running examples.
\bsub
\item $\comb{I} = \lam x.x$, representing the identity.
\item $\Om = \fix(\comb{I})$ representing the paradigmatic looping program.
\item $\mathbf{succ1} = \lam x.\succ (x)$, representing the successor function.
\item $\mathbf{succ2}=(\lam sn.s\cdot(s\cdot n))\cdot \mathbf{succ1}$, representing the function $f(x) = x+2$.
\item $\mathbf{add\_aux} = \lam f x y.\ifterm y {x} {(f \cdot(\succ x))\cdot (\pred y)}$, i.e.\ the functional
\[
	\Phi_f(x,y) = \begin{cases}
		x,&\textrm{if }y = 0,\\
		f(x+1,y-1),&\textrm{if }y > 0.\\
		\end{cases}
\]
\item $\mathbf{add} = \fix (\mathbf{add\_aux})$, i.e., the recursive definition of addition $f(x,y) = x+y$.
\esub
\end{example}

The operational semantics of \EPCF{} is defined via a call-by-name big-step (leftmost) weak reduction.

\begin{definition}\bsub
\item We let $\Val = \set{ \num n\st n\in\nat} \cup \set{\lam x.\esubst{M}{\sigma} \st M \textrm{ is an \EPCF{} term}}$ be the set of \EPCF{} \emph{values}.
\item 
The \emph{big-step weak reduction} is the least relation $\reddd$ from \EPCF{} terms to $\Val$, closed under the rules of Figure~\ref{fig:PCFheadDelay}.
\item We say that an \EPCF{} program $M$ is \emph{terminating} whenever $M\reddd V$ holds, for some $V\in\Val$.
Otherwise, we say that $M$ is a \emph{non-terminating}, or \emph{looping}, term.
\esub
\end{definition}

\begin{example} We show some of the terms from Example~\ref{ex:PCFterms}, at work.
\bsub
\item We have $\substseq{[\,]}{\mathbf{succ1}\cdot \mathbf{0}\reddd \num 1}$. 
To get the reader familiar with the operational semantics, we give the details:
\[
\infer[\betarule]{
	\substseq{[\,]}{(\lam x.\succ(x)) \cdot \num{0}} \reddd \num 1
}{
	\infer[\funrule]{
		\substseq{[\,]}{\lam x.\succ (x)}\reddd \lam x.\succ (x)
	}{
	}
	& \infer[\succrule]{
		\substseq{[x \leftarrow ([\,], \num{0})]}{\succ (x)} \reddd \num 1
	}{
		\infer[\varrule]{
			\substseq{[x \leftarrow ([\,], \num{0})]}{x} \reddd \num{0}
		}{
			\infer[\valrule]{\substseq{[\,]}{\num{0}} \reddd \num{0}}{}
		}
	}
}
\]
\item Similarly, $\substseq{[\,]}{\comb{I}\cdot \num 4} \reddd \num 4$, $\ \substseq{[\,]}{\comb{I\cdot I}} \reddd \comb{I}$, $\ \substseq{[\,]}{\mathbf{succ2}\cdot \num 1} \reddd \num 3$ and $\substseq{[\,]}{\mathbf{add}\cdot \num 5\cdot \num 1} \reddd \num 6$.
\item Since $\Om$ is looping, there is no $V\in\Val$ such that $\substseq{[\,]}{\Om}\reddd V$ is derivable.
\esub
\end{example}

We now endow \EPCF{} terms with a type system based on simple types.

\begin{definition}\label{def:simpletypes}\bsub
\item\label{def:simpletypes1} The set $\Types$ of \emph{(simple) types} over a \emph{ground type} $\tint$ is inductively defined by the grammar:
\begin{equation}\tag{$\Types$}
	\alpha,\beta\ ::=\ \tint \mid \alpha\to \beta
\end{equation}
The arrow associates to the right, in other words we write $\alpha_1\to\cdots\to\alpha_n\to\beta$ for \\$\alpha_1\to(\cdots\to(\alpha_n\to\beta)\cdots)$ ($= \vec\alpha\to\beta$, for short).
\item\label{def:simpletypes2} A \emph{typing context} $\Gamma$ is given by a set of associations between variables and types, written $x_1:\alpha_1,\dots,x_n : \alpha_n$.
In this case, we let $\dom(\Gamma) = \set{x_1,\dots,x_n}$.
When writing $\Gamma,x : \alpha$, we silently assume that $x\notin\dom(\Gamma)$.
\item \emph{Typing judgements} are triples, denoted $\Gamma\vdash^\E M : \alpha$, where $\Gamma$ is a typing context, $M$ is an \EPCF{} term and $\alpha\in\Types$.
\item \emph{Typing derivations} are finite trees built bottom-up in such a way that the root has shape $\Gamma\vdash^\E M : \alpha$ and every node is an instance of a rule from Figure~\ref{fig:typing}. In the rule $(\to_\mathrm{I})$ we assume wlog that $x\notin\Gamma$, by $\alpha$-conversion.
We also use an auxiliary predicate $\sigma \models \Gamma$ whose intuitive meaning is that $\Gamma$ is a typing context constructed from an explicit substitution $\sigma$.
\item 
When writing $\Gamma\vdash^\E M : \alpha$, we mean that this typing judgement is derivable.
\item We say that $M$ \emph{is typable} if $\Gamma\vdash^\E M : \alpha$ is derivable for some $\Gamma,\alpha$.
\esub
\end{definition}

\begin{figure*}[t]
$\arraycolsep=4pt\def\arraystretch{2.2}
\bar{c}
\bar{cccc}
\infer[(\mathrm{ax})]{\Gamma,x:\alpha\vdash^\E x:\alpha}{}
&
\infer[(0)]{\Gamma\vdash^\E\mathbf{0}:\tint}{}
&
\infer[(\mathrm{Y})]{\Gamma\vdash^\E \fix M:\alpha}{\Gamma\vdash^\E M : \alpha\rightarrow \alpha}
\\[1ex]
\infer[(+)]{\Gamma\vdash^\E\succ M:\tint}{\Gamma\vdash^\E M:\tint}
&
\infer[(-)]{\Gamma\vdash^\E\pred M:\tint}{\Gamma\vdash^\E M:\tint}
&
\infer[(\to_\mathrm{I})]{\Gamma\vdash^\E\lambda x.\esubst{M}{\sigma}:\alpha\rightarrow \beta}{\sigma \models \Delta & \Gamma,\Delta,x:\alpha\vdash^\E M : \beta}
\\[1ex]
\bar{ccc}
\infer[(\to_\mathrm{E})]{\Gamma\vdash^\E M\cdot N:\beta}{\Gamma\vdash^\E M:\alpha\rightarrow \beta &\Gamma\vdash^\E N:\alpha}
\ear
&&
\infer[(\mathrm{ifz})]{\Gamma\vdash^\E \ifterm LMN:\alpha}{
	\Gamma\vdash^\E L:\tint
	&
	\Gamma\vdash^\E M:\alpha
	&
	\Gamma\vdash^\E N:\alpha
}\\
\infer[(\sigma_0)]{[] \models \emptyset}{}
&\quad &
\infer[(\sigma)]{\sigma+[x\leftarrow(\rho, M)] \models \Gamma, x: \alpha}{\sigma \models \Gamma & \rho \models \Delta & \Delta \vdash^\E M : \alpha}
\ear
\ear
$
\caption{\EPCF{} type assignment system.}\label{fig:typing}
\end{figure*}

\begin{example} The following are examples of derivable typing judgments.
\bsub
\item {\footnotesize\[\begin{aligned}
\infer[(\to_\mathrm{I})]{
	\vdash^\E\lam x.\esubst{\succ (y)}{[y\leftarrow([\,], \mathbf{0})]}:\alpha\rightarrow \tint 
}{
	\infer[(\sigma)]{
		[y\leftarrow([\,], \mathbf{0})] \models y :\tint
	}{
		\infer[(\sigma_0)]{
			[\,]\models\emptyset
		}{
		}
		&
		\infer[(\sigma_0)]{
			[\,]\models\emptyset
		}{
		}
		& 
		\infer[(0)]{
			\vdash^\E \mathbf{0}:\tint
		}{
		}
	}
	\infer[(+)]{
		y:\tint,x:\alpha\vdash^\E\succ (y):\tint
	}{
	\infer[(\mathrm{ax})]{y:\tint,x:\alpha\vdash^\E y :\tint}{}
	}
}
\end{aligned}\]}
\item $\vdash^\E (\lam x.\succ (x))\cdot\mathbf{0}:\tint$.
\item $\vdash^\E (\lam sn.s\cdot (s\cdot n))\cdot(\lam x.\succ (x)):\tint\rightarrow\tint$.
\item $\vdash^\E \fix (\lam f x y.\ifterm y {x} {f \cdot(\succ x)\cdot (\pred y)}):\tint\rightarrow\tint\rightarrow\tint$.
\item $\vdash^\E \Om : \alpha$, for all $\alpha\in\Types$.
\esub
\end{example}

The following lemma summarizes the main (rather standard) properties of the language \EPCF.
\begin{lemma}\label{lem:epcftyping} Let $M$ be an \EPCF{} term, $V\in \Val$, $\alpha,\beta\in\Types$ and $\Gamma$ be a context.
\bsub
\item \label{lem:epcftyping1} (Syntax directedness) Every derivable judgement $\Gamma\vdash^\E M : \alpha$ admits a unique derivation.
\item \label{lem:epcftyping2} (Strengthening) If $\Gamma,x:\beta\vdash^\E M : \alpha$ and $x\notin\FV{M}$ then $\Gamma\vdash^\E M : \alpha$.
\item \label{lem:epcftyping4} (Subject reduction) For $M$ closed, $\vdash^\E M : \alpha$ and $\substseq{[\,]}{M \reddd V}$ entail $\vdash^\E V : \alpha$.
\esub
\end{lemma}
It follows that, if an \EPCF{} program $M$ is typable, then it is also typable in the empty context.
\section{Extended Addressing Machines}\label{sec:EAMS}
We extend the addressing machines from \cite{DellaPennaIM21} with instructions for performing arithmetic operations and conditional testing.
Natural numbers are represented by particular machines playing the role of numerals.

\subsection{Main definitions}

We consider fixed a countably infinite set $\Addrs$ of \emph{addresses} together with a distinguished countable subset $\mathbb{X} \subset \Addrs$, such that $\Addrs-\mathbb{X}$ remains infinite. 
Intuitively, $\mathbb{X}$ is the set of addresses that we reserve for the numerals, therefore hereafter we work under the hypothesis that $\mathbb{X} = \nnat[]$, an assumption that we can make without loss of generality.

Let $\Null\notin\Addrs$ be a ``null'' constant corresponding to an uninitialised register. 
Set $\Addrs_\Null = \Addrs\cup\set{\Null}$.

\begin{definition} 
\bsub
\item 
	An \emph{$\Addrs$-valued tape} $T$ is a finite ordered list of addresses $T = [a_1,\dots,a_n]$ with $a_i\in\Addrs$ for all $i\,(1\le i \le n)$. When $\Addrs$ is clear from the context, we simply call $T$ a tape. We denote by $\Tapes$ the set of all $\Addrs$-valued tapes.	
\item
	Let $a\in\Addrs$ and $T,T'\in\Tapes$. We denote by $\Cons a {T}$ the tape having $a$ as first element and $T$ as tail. We write $\appT{T}{T'}$ for the concatenation of $T$ and $T'$, which is an $\Addrs$-valued tape itself.

\item
	Given an index $i \ge 0$, an $\Addrs_\Null$-valued \emph{register} $R_i$ is a memory-cell capable of storing either $\Null$ or an address $a\in\Addrs$. We write $\val{R_i}$ to represent the value stored in the register $R_i$. (The notation $\oc R_i$ is borrowed from ML, where $\oc$ represents an explicit dereferencing operator.)
\item
	Given $\Addrs_\Null$-valued registers $R_0,\dots,R_{n}$ for $n\ge 0$, an address $a\in\Addrs$ and an index $i\ge 0$, we write $\vec R\repl{R_i}{a}$ for the list of registers $\vec R$ where the value of $R_i$ has been updated by setting $\val{R_i} = a$.
Notice that, whenever $i > n$, we assume that the contents of $\vec R$ remains unchanged, i.e.\ $\vec R\repl{R_i}{a} = \vec R$.
\esub
\end{definition}
Intuitively, the contents of the registers $R_0,\dots,R_n$ constitutes the \emph{state} of a machine, while the tape correspond to the list of its inputs.
The addressing machines from \cite{DellaPennaIM21} are endowed with only three instructions ($i,j,k,l$ range over indices of registers):
\begin{enumerate}[1.]
\item $\Load i$ : reads an address $a$ from the input tape, assuming it is non-empty, and stores $a$ in the register $R_i$.
If the tape is empty then the machine suspends its execution without raising an error.
\item $\Apply ijk$ : reads the addresses $a_1,a_2$ from $R_i$ and $R_j$ respectively, and stores in $R_k$ the address of the machine obtained by extending the tape of the machine of address $a_1$ with the address $a_2$. The resulting address is not calculated internally but rather obtained calling an external \emph{application map}.
\item $\Call i$ : transfers the computation to the machine having as address the value stored in $R_i$, whose tape is extended with the remainder of the current machine's tape.
\end{enumerate}
As a general principle, writing on a non-existing register does not cause issues as the value is simply discarded---this is in fact the way one can erase an argument. 
The attempt of reading an uninitialized register would raise an error---we however show that these kind of errors can be avoided statically (see Lemma~\ref{lem:correction}).

We enrich the above set of instructions with arithmetic operations mimicking the ones present in \PCF:
\begin{enumerate}[1.,resume]
\item $\Ifz i j k l$: implements the ``\emph{is zero?}'' test on $!R_i$. Assuming that the value of $R_i$ is an address $n\in\nat$, the instruction stores in $R_l$ the value of $R_j$ or $R_k$, depending on whether $n = 0$.
\item $\Pred i j$: if $\val{R_i} \in\nat$, the value of $R_j$ becomes $\val{R_i}\ominus1 = \max(\val{R_i} - 1, 0)$.
\item $\Succ i j$: if $\val{R_i} \in\nat$, then the value of $R_j$ becomes $\val{R_i}+1$.
\end{enumerate}
Notice that the instructions above need $R_i$ to contain a natural number to perform the corresponding operation. However, they are also supposed to work on addresses of machines that compute a numeral. 
For this reason, the machine whose address is stored in $R_i$ must first be executed, and only if the computation terminates with a numeral is the arithmetic operation performed. 
Clearly, if the computation terminates in an address not representing a numeral, then an error should be raised at execution time. We will see that these kind of errors can be avoided using a type inference algorithm (see Proposition~\ref{prop:typing}, below).

\begin{definition}\label{def:progs}
\bsub
\item \label{def:progs1}
	A \emph{program} $P$ is a finite list of instructions generated by the following grammar, where $\varepsilon$ represents the empty string and $i,j,k,l$ are indices of registers:
	\[
	\begin{array}{lcl}
	\ins{P}&::=&\Load i;\, \ins{P}\mid \ins{A}\\
	\ins{A}&::=&\Apply ijk;\, \ins{A}\mid \Ifz ijkl;\, \ins{A}\mid
	\Pred ij;\, \ins{A}\mid\Succ ij;\, \ins{A}\mid\ins{C}\\
	\ins{C}&::=&\Call i \mid \varepsilon
	\end{array}
	\]
	Thus, a program starts with a list of $\ins{Load}$'s, continues with a list of $\ins{App}$, $\ins{Test}$, $\ins{Pred}$, $\ins{Succ}$, and possibly ends with a $\ins{Call}$. Each of these lists may be empty, in particular the empty program $\varepsilon$ can be generated.
	\item In a program, we write $\Load (i_1, \dots, i_n)$ as an abbreviation for the instructions $\Load i_1;\,\cdots;\,\Load i_n$. 

\item\label{def:progs2}
	Let $P$ be a program, $r\ge 0$, and $\cI\subseteq \set{0,\dots,r-1}$ be a set of indices corresponding to the indices of initialized registers. 
	Define the relation $\cI\models^{r} P$, whose intent is to specify that $P$ does not read uninitialized registers, as the least relation closed under the rules:
\[
	\bar{cccc}
		\infer{\cI\models^{r}\varepsilon}{}
		&
		\infer{\cI\models^{r}\Call i}{i\in \cI}
		&
		\infer{\cI\models^{r} \Pred ij;\, \ins{A}}{\cI\cup\set{j}\models^{r}  \ins{A} & i\in \cI& j<r}
		\\[1ex]
		\infer{\cI\models^{r} \Load i;\, \ins{P}}{\cI\cup\set{i}\models^{r}  \ins{P} & i< r}		
		&
		\infer{\cI\models^{r} \Load i;\, \ins{P}}{\cI\models^{r}  \ins{P} & i\ge r}\qquad
		&
	\infer{\cI\models^{r} \Succ ij;\, \ins{A}}{\cI\cup\set{j}\models^{r}  \ins{A} & i\in \cI& j<r}
		\\[1ex]
	\infer{\cI\models^{r} \Apply ijk;\, \ins{A}}{\cI\cup\set{k}\models^{r}  \ins{A} & i,j\in \cI& k<r}	
	&
	\multicolumn{2}{c}{		
	\infer{\cI\models^{r} \Ifz ijkl;\, \ins{A}}{\cI\cup\set{l}\models^{r}  \ins{A} & i,j,k\in \cI& l<r}}
	\ear
\]
\item\label{def:progs3} 
	A program $P$ is \emph{valid with respect to $R_0,\dots,R_{r-1}$} if $\cR\models^{r} P$ holds for
	$\cR = \set {i\st R_i \neq\Null \und 0\le i < r}$.
\esub
\end{definition}

\begin{example} For each of these programs, we specify its validity with respect to $R_0 = 7, R_1 = a, R_2 = \Null$ (i.e., $r = 3$).
\[
\bar{lcr}
P_1 = \Pred 02;\,\Call 2&&\textrm{(valid)}\\
P_2 = \Load (2,8);\,\Ifz 0120;\, \Call 0&&\textrm{(valid)}\\
P_3 = \Load (0, 2,8);\, \Call 8&\qquad&\textrm{(calling the uninitialized register $R_8$, thus not valid)}
\ear
\]
\end{example}

\begin{lemma}\label{lem:correction}
Given $\Addrs_\Null$-valued registers $\vec R$ and a program $P$ it is decidable whether $P$ is valid w.r.t.\ $\vec R$.
\end{lemma}

\begin{proof} Decidability follows from the syntax directedness of Definition~\ref{def:progs}\eqref{def:progs2}, and the preservation of the invariant $\cI\subseteq\set{0,\dots,r-1}$, since $\cI$ is only extended with $k<r$.
\end{proof}

\begin{definition}\label{def:AM}
\bsub
\item
	An \emph{extended addressing machine} (\emph{EAM}) $\mM$ with $r$ registers over $\Addrs$ is given by a tuple:
\[
	\mM = \tuple{R_0,\dots,R_{r-1},P,T}
\] 
where $\vec R$ are $\Addrs_\Null$-valued registers, $P$ is a program valid w.r.t.\ $\vec R$ and $T\in\Tapes$ is an (input) tape.
\item
	We write $\mM.r$ for the number of registers of $\mM$, $\mM.R_i$ for its $i$-th register, $\mM.P$ for the associated program and $\mM.T$ for its input tape.
	When writing ``$R_i = a$'' in a tuple we indicate that $R_i$ is present and $\val{R_i} = a$.
\item
	We say that an extended addressing machine $\mM$ as above is \emph{stuck}, written $\stuck{\mM}$, whenever its program has shape $\mM.P = \Load i;P$ but its input-tape is empty $\mM.T = []$. Otherwise $\mM$ is \emph{ready}, written $\lnot\stuck{\mM}$.
\item
	The set of all extended addressing machines over $\Addrs$ will be denoted by $\cM$.
\item For $n\ge 0$, the \emph{$n$-th numeral machine} is defined $\mach{n}= \tuple{R_0,\varepsilon,[]}$ with $\val{R_0} = n$.
\item\label{def:AM6} For $n\ge 0$ and $a\in\Addrs$, define 
\[
	\mYn{n}^{a} = \tuple{(R_0 = a,R_1=\Null,\dots,R_{n+1} = \Null,P,[]}
\] 
where$\bar[t]{rl}
P =&  \Load (1,\dots,{n+1});\Apply  0 10;\cdots;
	 \Apply 0{n+1}0;\Apply  1 21;\cdots; \\
	&\Apply 1{n+1}1;\Apply 101; \Call 1
\ear
$
\esub
\end{definition} 

We now enter into the details of the addressing mechanism which constitutes the core of this formalism.

\begin{definition} \label{def:bijectivelookup+Y} Recall that $\nat$ stands for an infinite subset of $\Addrs$, here identified with the set of natural numbers, and $\mYn{n}^a$ has been introduced in Definition~\ref{def:AM}(\ref{def:AM6}).
\begin{enumerate}
\item Since $\cM$ is countable, we can fix a bijective function $\Lookup : \cM \to  \Addrs$ satisfying the following conditions:
	\bsub
	\item (Numerals)
	$\forall n \in \mathbb{N}\,.\, \Lookup{\mach{n}} = n$, where $\mach{n}$ is the $n$-th numeral machine;
\item 
	(Fixed point combinator) for all $n\ge 0$, there exists an address $a\in\Addrs-\nnat[]$ such that $\Lookup(\mYn{n}^{a}) = a$.
\esub
We say that the bijection $\Lookup(\cdot)$ is an \emph{address table map} and call the element $\Lookup \mM$ the {\em address of the EAM $\mM$}.
We simply write $\mYn{n}$ for the machine satisfying the equation above and $a_{\mYn{n}}$ for its address, i.e.~$\Lookup(\mYn{n}) = a_{\mYn{n}}$.
\item
	For $a\in\Addrs$, we write $\Lookinv{a}$ for the unique machine having address $a$, i.e., 
	$\Lookinv{a} = \mM\iff \Lookup\mM = a.$
\item 
	Given $\mM\in\cM$ and $T'\in\Tapes$, we write $\appT{\mM}{T'}$ for the machine $\tuple{\mM.\vec R,\mM.P,\appT{\mM.T}{T'}}$.
\item
	Define the \emph{application map} $(\App{}{}) : \Addrs\times\Addrs\to \Addrs$ by setting $\App{a}{b} = \Lookup (\append{\Lookinv{a}}{b})$, i.e., the \emph{application} of $a$ to $b$ is the unique address $c$ of the EAM obtained by adding $b$ at the end of the input tape of the EAM $\Lookinv{a}$.	  	
\esub
\end{definition}

\begin{example}\label{ex:extabsmach} The following are examples of EAMs (whose registers are assumed uninitialized, i.e.\ $\vec R = \vec \Null$). 
\bsub
\item $\mach{Succ1} := \tuple{R_0, \Load 0; \Succ 0 0; \Call 0, []}$.
\item $\mach{Succ2} := \tuple{R_0, R_1, \Load 0;\Load 1; \Apply 0 1 1; \Apply 0 1 1; \Call 1,[a_{\mach{S}}]}$, where $a_{\mach{S}} = \Lookup{\mach{Succ1}}$.
\item $\mach{Add\_aux} := \tuple{\vec R,P, []}$ with $\mach{Add\_aux}.r = 5$ and $P = \Load (0, 1, 2); \Pred 1 3; \Succ 2 4;\Apply 0 3 0;$ $\Apply 0 4 0; \Ifz 1 2 0 0; \Call 0$.
\esub
\end{example}

\begin{remark}\label{rem:blackhole} 
In general, there are uncountably many possible address table maps of arbitrary computational complexity. 
A natural example of such maps is given by \emph{G\"odelization}, which can be performed effectively.
The framework is however more general and allows to consider non-r.e.\ sets of addresses like the complement $K^c$ of the halting set
\[
	K = \set{(i, x) \st \textrm{the program $i$ terminates when run on input $x$}}
\]
and a non-computable function $\Lookup: \mathcal{M}_{K^c} \to K^c$ as a map.
\end{remark} 
In an implementation of EAMs the address table map should be computable---one can choose a fresh address from $\Addrs$ whenever a new machine is constructed, save the correspondence in some table and retrieve it in constant time.

\begin{remark}\label{rem:blackhole} 
Depending on the chosen address table map, it might be possible to construct infinite (static) chains of EAMs $(\mM)_{n\in\nat}$, e.g., 
$\mM_n = \tuple{R_0 = \Lookup\mM_{n+1},\varepsilon,[]}$.
\end{remark} 
\noindent The results we present are independent from the choice of $\Lookup$.
\subsection{Operational semantics}

The operational semantics of extended addressing machines is given through a small-step rewriting system. 
The reduction strategy is deterministic, since the only applicable rule at every step is univocally determined by the first instruction of the internal program, the contents of the registers and the head of the tape.

\begin{definition} We introduce a fresh constant $\err\notin\cM$ to represent a machine raising an error.
\bsub
\item Define a reduction strategy $\redh$ on EAMs, representing one step of computation, as the least relation $\redh\ \subseteq\cM\times(\cM\cup\set{\err})$ closed under the rules in Figure~\ref{fig:am:small_step}.
\item The \emph{multistep reduction} $\reddh$ is defined as the transitive-reflexive closure~of~$\redh$.
\item Given $\mM, \mN, \mM\reddh\mN$, we write $|\mM\reddh\mN|\in\nat$ for the length of the (unique) reduction path from $\mM$ to $\mN$.
\item For $\mM,\mN\in\cM$, we write $\mM\convh\mN$ if they have a common reduct $\mach{Z}\in\cM\cup\set{\err}$, i.e.\ $\mM\reddh \mach{Z} {~}_{\mach{c}}\!\!\twoheadleftarrow\mN$.
\item An extended address machine $\mM$: \emph{is in final state} if it cannot reduce, written $\mM\not\redh$; \emph{reaches a final state} if $\mM\reddh\mM'$ for some $\mM'\in\cM$ in final state; \emph{raises an error} if $\mM\reddh \err$;  \emph{does not terminate}, otherwise.
\esub
\end{definition}
\begin{figure*}[t!]
\centering
{\bf Unconditional rewriting rules}
\[	\bar{rcl}
	\tuple{\vec R,\Call i; P,T}&\redh&\appT{\Lookinv{\val{R_i}}}{T}\\
	\tuple{\vec R,\RaS i;P,\Cons a{T}} &\redh& \tuple{\vec R[R_i := a],P,T}\\
	\tuple{\vec R,\Apply i j k; P,T}&\redh&\tuple{\vec R[R_k := \App{\val{R_i}}{\val{R_j}}],P,T}\\
	\ear
\]
{\bf Under the assumption that $\Lookinv{\val{R_i}}\not\redh$ (i.e., it is in final state).}
\[
	\bar{rcl}
	\tuple{\vec R,\Pred i j;P,T} &\redh& 
	\begin{cases}
	\tuple{\vec R[R_j := \val{R_i} \ominus 1,P,T},&\textrm{if }\val{R_i} \in \nat,\\
	\err,&\textrm{otherwise}.
	\end{cases}
	\\[3ex]
	\tuple{\vec R,\Succ i j;P,T} &\redh& 
	\begin{cases}
	\tuple{\vec R[R_j := \val{R_i} \oplus 1,P,T},&\textrm{if }\val{R_i} \in \nat,\\
	\err,&\textrm{otherwise}.
	\end{cases}
	\\[3ex]	
	\tuple{\vec R,\Ifz i j k l; P, T}&\redh&\begin{cases}
		\tuple{\vec R[R_l := \val{R_j}], P, T},&\textrm{if }\val{R_i} = 0,\\[3pt]
		\tuple{\vec R[R_l := \val{R_k}], P, T},&\textrm{if }\val{R_i} \in \nnat,\\
		\err,&\textrm{otherwise}.
		\end{cases}\\[3ex]
	\ear
\]
{\bf Under the assumption that $\Lookinv{\val{R_i}}\redh \mach A$ (i.e., it is not in final state).}
\[
	\bar{rcl}
	\tuple{\vec R,\Pred i j;P,T} &\redh& 
	\tuple{\vec R[R_i := \Lookup{\mach A}],\Pred i j;P,T}
	\\[1ex]
	\tuple{\vec R,\Succ i j;P,T} &\redh& 
	\tuple{\vec R[R_i := \Lookup{\mach A}],\Succ i j;P,T}
	\\[1ex]	
	\tuple{\vec R,\Ifz i j k l; P, T}&\redh&\tuple{\vec R[R_i := \Lookup{\mach A}],\Ifz i j k l; P, T}\\[1ex]
	\ear
\]\caption{Small-step operational semantics for extended addressing machines.}
\label{fig:am:small_step}
\end{figure*}

Notice that since the redexes in Figure~\ref{fig:am:small_step} are not overlapping, the confluence of $\reddh$ follows easily (cf.\ \cite[Lemma~2.11(2)]{DellaPennaIM21}).
\begin{lemma}\label{lem:reduction}
If $\mM \reddh \mach{M'}$, then $\appT{\mM}{\Lookup\mN} \reddh \appT{\mach{M'}}{\Lookup\mN}$.
\begin{proof}
By induction on the length of $\mM \reddh \mach{M'}$.
\end{proof}
\end{lemma}

\begin{example} See Example~\ref{ex:extabsmach} for the definition of $\mach{Succ1}$, $\mach{Succ2}$, $\mach{Add\_aux}$. 
\bsub
\item We have $\appT{\mach{Succ1}}{[0]}\reddh \mach 1$ and $\appT{\mach{Succ2}}{[1]}\reddh \mach 3$.
\item Define $\mach{Add} = \appT{\mYn{0}}{[\Lookup \mach{Add\_aux}]}$, an EAM performing the addition. We show:
\[
	\bar{l}
	\appT{\mach{Add}}{[1,3]}
\redh \Tuple{(R_0 = a_{\mYn{0}},R_1 = \Lookup\mach{Add\_aux}), \Apply 0 1 0;\Apply 1 01; \Call 1,[1, 3]}\\[1ex]
\reddh \Tuple{\vec R, \Load (0,1,2); \Pred 1 3; \Succ 2 4; \Apply 0 3 0;\Apply 0 4 0;\\\Ifz 1 2 0 0; \Call 0, [\Lookup{\mach{Add}}, 1, 3]}\\[2ex]
	\reddh\Tuple{R_0 = \Lookup{\mach{Add}}, R_1 = 1, R_2 = 3, R_3, R_4, \Pred 1 3; \Succ 2 4; \Apply 0 3 0;\\\Apply 0 4 0; \Ifz 1 2 0 0; \Call 0, []}\\[1ex]
	\reddh\Tuple{R_0 = \Lookup(\appT{\mach{Add}}{[0,4]}), R_1 = 1, R_2 = 3, R_3 = 0, R_4 = 4, \Ifz 1 2 0 0; \Call 0, []}\\[1ex]
	\reddh\Tuple{R_0 = \Lookup(\appT{\mach{Add}}{[0,5]}), R_1 = 0, R_2 = 4, R_3 = 0, R_4 = 5, \Ifz 1 2 0 0; \Call 0, []}
	\reddh \mach{4}
\ear
\]
\item For $\mach{I} = \tuple{R_0 = \Null,\Load 0;\Call 0,[]}$,  $\appT{\mYn{0}}{[\Lookup\mach{I}]}\reddh\appT{\mach{I}}{[\Lookup(\appT{\mYn{0}}{[\Lookup\mach{I}]})]}$.
\item $\appT{\mYn{n}}{[\Lookup\mach M,d_1,\dots,d_n]} \reddh \appT{\mM}{[d_1,\dots,d_n,\Lookup(\appT{\mYn{n}}{[\Lookup\mM,d_1,\dots,d_n]})]},$ for all $n\ge 0$, $\mM\in\cM$, $\vec d\in\Addrs$.
\esub
\end{example}

\section{Typing Algorithm}\label{sec:Types}
\begin{figure*}[t!]
$
\bar{c}
\bar{ccc}
\infer[\mathrm{nat}]{\vdash\mM : \tint}{\Lookup\mM \in \nat}&\qquad\qquad
	\infer[\mathrm{fix}_n]{\vdash \mM : (\vec\delta\to\alpha \to \alpha)\to\vec\delta \to \alpha}{\Lookup{\mM} = a_{\mYn{n}}& \vec\delta = \delta_1\to\cdots\to\delta_n}\qquad\qquad
	&
	\infer[R_{()}]{\Delta\vdash\tuple{(),P,T} : \alpha}{\Delta\Vdash(P,T) : \alpha}
\ear\\[1ex]
	\bar{cc}	
	\infer[R_\Null]{\Delta\vdash\tuple{(R_0,\dots,R_{r}),P,T} : \alpha}{\Delta\vdash\tuple{R_0,\dots,R_{r-1},P,T} : \alpha& \val{R_r} = \Null}\qquad
	&
	\infer[R_\Types]{\Delta\vdash\tuple{(R_0,\dots,R_{r}),P,T} : \alpha}{R_r : \beta,\Delta\vdash\tuple{R_0,\dots,R_{r-1},P,T} : \alpha& \vdash \Lookinv{\val{R_r}} : \beta}
	\\[1ex]
	\infer[\mathrm{load_{\Null}}]{\Delta \Vdash (\Load i;P,[]) : \beta\to\alpha}{\Delta[R_i : \beta] \Vdash (P,[]) : \alpha}
	&
	\infer[\mathrm{load_{\Types}}]{\Delta\Vdash (\Load i;P, \Cons a {T}) : \alpha}{\Delta[R_i : \beta] \Vdash (P,T) : \alpha &\vdash \Lookinv {a} : \beta}
	\\[1ex]
	\infer[\mathrm{pred}]{\Delta, R_i : \tint \Vdash (\Pred i j;P,T) : \alpha}{(\Delta, R_i : \tint)[R_j : \tint] \Vdash (P,T) : \alpha}
	&
	\infer[\mathrm{succ}]{\Delta, R_i : \tint \Vdash (\Succ i j;P,T) : \alpha}{(\Delta, R_i : \tint)[R_j :\tint] \Vdash (P,T) : \alpha}
	\\[1ex]
	\ear
	\\[1ex]
		\infer[\mathrm{test}]{\Delta, R_i : \tint,  R_j :\beta, R_k : \beta \Vdash (\Ifz i j k l;P,T) : \alpha}{
		(\Delta, R_i : \tint,  R_j : \beta, R_k : \beta)[R_l : \beta] \Vdash (P,T) : \alpha 
		}\\[1ex]
		\infer[\mathrm{app}]{\Delta, R_i : \alpha \to \beta, R_j : \alpha \Vdash (\Apply i j k;P,T) : \delta}{
		(\Delta, R_i : \alpha \to \beta, R_j : \alpha)[R_k : \beta] \Vdash (P,T) : \delta}\\[1ex]
	\infer[\mathrm{call}]{\Delta, R_i : \alpha_1\to\dots\to\alpha_n\to\alpha \Vdash (\Call i,[\Lookup{\mM_1},\dots,\Lookup{\mM_n}]) : \alpha}{
		\vdash \mM_1 : \alpha_1 & \cdots & \vdash \mM_n : \alpha_n}
	\ear
$
\caption{Typing rules for extended addressing machines.}\label{fig:eamstyping}
\end{figure*}

Recall that the set $\Types$ of (simple) types has been introduced in Definition~\ref{def:simpletypes}\eqref{def:simpletypes1}. 
We now show that certain EAMs can be typed, and that typable machines do not raise error during their execution.

\begin{definition}\label{def:typing} 
\bsub
\item A \emph{typing context} $\Delta$ is a finite set of associations between registers and types, represented as a list $R_{i_1} : \alpha_1,\dots,R_{i_{n}} : \alpha_{n}$. The indices $i_1,\dots,i_n$ are not necessarily consecutive.
\item We denote by $\Delta[R_i : \alpha]$ the typing context $\Delta$ where the type associated with $R_i$ becomes $\alpha$. If $R_i$ is not present in $\Delta$, then $\Delta[R_i : \alpha] = \Delta,R_i : \alpha$.
\item Let $\Delta$ be a typing context, $\mM\in\cM$, $P$ be a program, $T\in\Tapes$ and $\alpha\in\Types$. We define the typing judgements
\[
	\Delta \vdash \mM : \alpha\qquad\qquad\qquad	\Delta \Vdash (P,T) : \alpha
\] 
by mutual induction as the least relations closed under the rules of Figure~\ref{fig:eamstyping}.
The rules $(\mathrm{nat})$ and $(\mathrm{fix})$ are the base cases and take precedence over $(R_\Null)$ and $(R_{\Types})$.
\item For $R_{i_1},\dots,R_{i_n}\in\vec R$, write $R_{i_1}:\beta_{i_1},\dots,R_{i_n}:\beta_{i_n}\models \vec R$ if $\Lookinv{\val{R_j}} : \beta_j$, for all $j\in\set{i_1,\dots,i_n}$.
\esub
\end{definition}

The algorithm in Figure~\ref{fig:eamstyping} deserves some discussion. As it is presented as a set of inference rules, one should reason bottom-up.
To give a machine $\mM$ a type $\alpha$, one needs to derive the judgement $\vdash \mM : \alpha$. The machines $\mach{n}$ and $\mYn{n}$ are recognizable from their addresses and the rules $(\mathrm{nat})$ and $(\mathrm{fix})$ can thus be given higher precedence. Otherwise, the rule $(R_\Types)$ allows to check whether the value in a register is typable and only retain its type, the rule $(R_\Null)$ allows to get rid of uninitialized registers.
Once this initial step is performed, one needs to derive a judgement of the form $R_{i_1}:\beta_{i_1},\dots,R_{i_n}:\beta_{i_n}\Vdash (P,T) : \alpha$, where $P$ and $T$ are the program and the input tape of the original machine respectively. This is done by verifying the coherence of the instructions in the program with the types of the registers and of the values in the input tape.
As a final consideration, notice that the rules in Figure~\ref{fig:eamstyping} can only be considered as an algorithm when the address table map is effectively given.
Otherwise, the algorithm would depend on an oracle deciding $a = \#\mM$.\\[-4ex]
\begin{remark}\label{rem:abouttypings}
\bsub
\item\label{rem:abouttypings1} For all $\mM\in\cM$ and $\alpha\in\Types$, we have $\vdash \mM : \alpha$ if and only if there exists $a\in\Addrs$ such that both $\Lookinv{a} : \alpha$ and $ \Lookup{\mM} = a$ hold.
\item\label{rem:abouttypings2} 
If $\Lookup\mM\notin \nnat[]\cup\set{a_{\mYn{n}}\st n\ge 0}$, then 
$\vdash \mM : \alpha \iff \exists\Delta\,.\, [\Delta\models \mM.\vec R\ \land\ \Delta \Vdash (\mM.P,\mM.T) : \alpha ]$
\item\label{rem:abouttypings3} The higher priority assigned to the rules $(\mathrm{nat})$ and $(\mathrm{fix})$ does not modify the set of typable machines, rather guarantees the syntax-directedness of the system.
\esub
\end{remark}

\begin{example} The following typing judgements are derivable.
\bsub
\item $\vdash \mach{Succ2} : \tint \rightarrow \tint$
\item $\vdash \mach{Add} : \tint\rightarrow\tint\rightarrow\tint$, where $\mach{Add} = \appT{\mYn{0}}{[\Lookup\mach{Add\_aux}]}$
\item For a smaller example, like $\vdash \mach{Succ1} : \tint\to\tint$, we can provide the whole derivation tree:
{\footnotesize\[
\infer[R_\Null]{\vdash \tuple{(R_0=\Null), \Load 0; \Succ 0 0; \Call 0, []} : \tint\to\tint}{\infer[R_{()}]{\vdash\tuple{(), \Load 0; \Succ 0 0; \Call 0, []} : \tint\to\tint}{\infer[\mathrm{load_{\Null}}]{\Vdash \tuple{\Load 0; \Succ 0 0; \Call 0, []} : \tint \to \tint}{\infer[\mathrm{succ}]{R_0 : \tint \Vdash \tuple{\Succ 0 0; \Call 0, []} : \tint}{\infer[\mathrm{call}]{R_0 : \tint\Vdash \tuple{\Call 0, []} : \tint}{}	}}
	}& \val{R_0} = \Null}
\]
}
\esub
\end{example}

\begin{lemma} Let $\mM\in\cM$, $\alpha\in\Types$. Assume that $\# : \mM\to\Addrs$ is effectively given.
\bsub
\item If $\mM = \tuple{\vec R = \Null,P,[]}$ then the typing algorithm is capable of deciding whether $\vdash \mM : \alpha$ holds.
\item In general, the typing algorithm semi-decides whether $\vdash \mM : \alpha$ holds.
\esub
\end{lemma}

\begin{proof}(Sketch) (i) In this case, $\vdash \mM : \alpha$ holds if and only if $\Vdash (\mM.P,[])$ does. By induction on the length of $\mM.P$, one verifies if it is possible to construct a derivation. Otherwise, conclude that $\vdash \mM : \alpha$ is not derivable.

(ii) In the rules $(R_\Types)$ and $(\mathrm{load}_\mathbb{T})$, one needs to show that a type for the premises exists.
As the set of types is countable, and effectively given, one can easily design an algorithm constructing a derivation tree (by dovetailing).
However, the algorithm cannot terminate when executed on $\mM_0$ from Remark~\ref{rem:blackhole}.
\end{proof}

The machine $\mM_0$ in Remark~\ref{rem:blackhole} cannot be typable because it would require an infinite derivation tree.

\begin{proposition}\label{prop:typing} 
Let $\mM,\mach {M'},\mN,\in\cM$ and $\alpha,\beta\in\Types$.
\bsub
\item\label{prop:typing1}  If $\vdash \mM : \beta \to \alpha$ and $\vdash\mN : \beta$ then $\vdash\appT{\mM}{[\Lookup{\mN}]} : \alpha$.
\item\label{prop:typing2}  If $\vdash \mM : \alpha$ and $\mM\redh \mN$ then $\vdash \mN : \alpha$.
\item\label{prop:typing4}  If $\vdash \mM : \tint$ then either $\mM$ does not terminate or $\mM\reddh \mach{n}$, for some $n\ge 0$.
\item\label{prop:typing3}  If $\vdash \mM : \alpha$ then $\mM$ does not raise an error.
\esub
\end{proposition}

\begin{proof} (i) Simultaneously, one proves that $\Delta\Vdash (P,T) : \beta\to\alpha$ and $\vdash \mN : \beta$ imply $\Delta\Vdash(P,\appT{T}{[\Lookup \mN]}) : \alpha$.
Proceed by induction on a derivation of $\vdash \mM : \beta\to\alpha$ (resp.\ $\Delta\Vdash (P,T) : \beta\to\alpha$). 

Case ($\mathrm{nat}$) is vacuous.

Case ($\mathrm{fix_n}$). We show the case for $n = 0$, the others being similar. By definition of $\mYn{0}$, we have:
\[
	\appT{\mYn{0}}{[\Lookup\mN]} = \Tuple{(\Null,a_{\mYn{0}}), \Load 0 ; \Apply 1 0 1;\Apply 0 1 0; \Call 0,[\Lookup{\mN}]}.
\]
Notice that, in this case, $\beta = \alpha\to\alpha$. 
Using $\Lookinv{a_{\mYn{0}}} = \mYn{0}$, we derive:
{\footnotesize
\[
	\infer[R_\mathbb{T}]{\vdash\tuple{(R_0 = \Null,R_1 = a_{\mYn{0}}), \Load 0 ; \Apply 1 0 1;\Apply 0 1 0; \Call 0,[\Lookup{\mN}]} : \alpha}{\infer=[R_{()};\,R_\emptyset]{
	R_1 : (\alpha\to\alpha)\to\alpha\Vdash \tuple{R_0=\Null,\Load 0 ;\cdots,[\Lookup{\mN}]} : \alpha
	}{
	\infer[\mathrm{load}_\mathbb{T}]{R_1 : (\alpha\to\alpha)\to\alpha\Vdash (\Load 0 ;\cdots,[\Lookup{\mN}]) : \alpha}{
		\infer=[\mathrm{app};\mathrm{app}]{R_0 : \alpha\to\alpha,R_1 : (\alpha\to\alpha)\to\alpha\Vdash (\Apply 1 0 1;\cdots,[]) : \alpha}{
				\infer[\mathrm{call}]{R_0 : \alpha,R_1 :\alpha\Vdash ( \Call 0,[]) : \alpha}{}
		}
		&
		\vdash \mN : \alpha\to\alpha
		}
	}
		&
		\infer{\vdash \mYn{0} : (\alpha\to\alpha)\to\alpha}{\mathrm{fix}_0}
	}
\]
}
\indent Case $\mathrm{load}_\Null$. Then $P = \Load i;P'$, $T = []$ and $\Delta[R_i : \beta] \Vdash (P',[]) : \alpha$. 
By assumption $\vdash \mN : \beta$, so we conclude $\Delta \Vdash (\Load i;P',[]) : \alpha$ by applying $\mathrm{load}_\Types$.
All other cases derive straightforwardly from the IH.

(ii) The cases $\mM = \mYn{n}$ or $\mM = \mach{n}$ for some $n\in\nat$ are vacuous, as these machines are in final state.
Otherwise, by Remark~\ref{rem:abouttypings}\eqref{rem:abouttypings2},  $\Delta \Vdash (\mM.P,\mM.T) : \alpha$ for some $\Delta\models \mM.\vec R$.
By cases on the shape of $\mM.P$. 

Case $P = \Load i;P'$. Then $\mM.T = \Cons a T'$ otherwise $\mM$ would be in final state, and $\mN = \tuple{\vec R[R_i := a],P',T'}$. 
From $(\mathrm{Load}_\Types)$ we get $\Delta[R_i : \beta] \Vdash (P',T') : \alpha$ for some $\beta\in\Types$ satisfying $\Lookinv {a} : \beta$. 
As $\Delta\models \vec R$ we derive $\Delta[R_i : \beta]\models \vec R[R_i := a]$, so as $N = \Tuple{\vec R[R_i := a],P',T'}$, by Remark~\ref{rem:abouttypings}\eqref{rem:abouttypings2}, $\vdash N : \alpha$.

Case $P = \Call i$. Then $R_i : \alpha_1\to\cdots\to\alpha_n\to\alpha$, $T = [\Lookup{\mM_1},\dots,\Lookup{\mM_n}]$ and $\vdash \mM_j : \alpha_j$, for all $j\le n$.
In this case, $\mN = \appT{\Lookinv{\val{(\mM.R_i)}}}{T}$ with $\vdash \Lookinv{\val{(\mM.R_i)}} : \alpha_1\to\cdots\to\alpha_n\to\alpha$, so we conclude by (i).

All other cases follows easily from the IH.

  (iii) Assume that $\vdash \mM : \tint$ and $\mM\reddh \mN$ for some $\mN$ in final state. By (ii), we obtain that $\vdash \mN: \tint$ holds, therefore $\mN= \mach{n}$ since numerals are the only machines in final state typable with $\tint$.

 (iv) The three cases from Figure~\ref{fig:am:small_step} where a machine can raise an error are ruled out by the typing rules ($\mathrm{pred}$), ($\mathrm{succ}$) and ($\mathrm{test}$), respectively. Therefore, no error can be raised during the execution.
\end{proof}

\section{Translation and Simulation}\label{sec:trans}

We define a type-preserving translation from \EPCF{} terms to extended addressing machines. 
More precisely, we show that if $\Gamma\vdash^\E M : \alpha$ is derivable then $M$ is transformed into a machine $\mM$ which is typable with the same $\alpha$.
By Proposition~\ref{prop:typing}, $\mM$ never raises a runtime error and well--typedness is preserved during its execution.
We then show that if a well-typed \EPCF{} program $M$ computes a value $\num n$, then its translation $\mM$ reduces to the corresponding EAM~$\mach{n}$.
Finally, this result is transported to \PCF{} using their equivalence on programs of type $\tint$.

We start by showing that EAMs implementing the main $\PCF{}$ instructions are definable.
We do not need any machinery for representing explicit substitutions because they are naturally modelled by the evaluation strategy of EAMs.

\begin{lemma}\label{lem:existenceofeams} 
Let $n\ge 0$. There are EAMs satisfying (for all $a,b,c,d_1,\dots,d_n\in\Addrs$):
\bsub
\item\label{lem:existenceofeams1} $\appT{\mProj{n}{i}}{[d_1,\dots,d_n]}\hspace{25pt} \reddh d_i$, for $1\le i \le n$;
\item\label{lem:existenceofeams2}  
	$\appT{\mAppn{n}}{[a,b,d_1,\dots,d_n]} \reddh \appT{\Lookinv{a}}{[d_1,\dots,d_n,b\cdot d_1\cdots d_n]}$;
\item\label{lem:existenceofeams3}  
	 $\appT{\mPredn{n}}{[a,d_1,\dots,d_n]}\reddh \tuple{R_0 = a\cdot d_1\cdots d_n,\vec R,;\Pred00;\Call 0,[]}$;
\item\label{lem:existenceofeams4}   
	$\appT{\mSuccn{n}}{[a,d_1,\dots,d_n]}\reddh\tuple{R_0 = a\cdot d_1\cdots d_n,\vec R,;\Succ00;\Call 0,[]}$;
\item\label{lem:existenceofeams5}
	$\appT{\mIfZn{n}}{[a,b,c,d_1,\dots,d_n]}\reddh\tuple{R_0 = a\cdot \vec d,R_1 = b\cdot \vec d,R_2 = c\cdot \vec d,\vec R, \Ifz 0120;\Call 0,[]}$.
\esub
\end{lemma}

\begin{proof} Easy. As an example, we give a possible definition of the predecessor:
\[
\mPredn{n} = \Tuple{R_0,\dots,R_n, \Load(0,\dots, n);\Apply 010;\cdots;\Apply 0n0;\Pred00;\Call 0,[]}
\]
The others are similar.
\end{proof}

\begin{lemma}\label{lem:welltypedeams}  
The EAMs in the previous lemma can be defined in order to ensure their typability (for all $n\ge0$, $\alpha,\beta,\gamma,\delta_i\in\Types$):
\bsub
\item\label{lem:welltypedeams1}
	$\vdash \mProj{n}{i} : \vec \delta\to\delta_i$, with $\vec \delta = \delta_1\to\cdots\to\delta_n$;
\item\label{lem:welltypedeams2} 
	$\vdash\mAppn{n} : (\vec \delta\to\beta\to\alpha)\to(\vec \delta\to\beta)\to\vec\delta\to\alpha$;
\item\label{lem:welltypedeams3}
	$\vdash \mPredn{n} : (\vec \delta\to\tint)\to\vec \delta\to\tint$;
\item\label{lem:welltypedeams4} 
	$\vdash\mSuccn{n} : (\vec \delta\to\tint)\to\vec \delta\to\tint$;
\item\label{lem:welltypedeams5} 
	$\vdash\mIfZn{n} : (\vec \delta\to\tint)\to(\vec \delta\to\alpha)\to(\vec \delta\to\alpha)\to\vec \delta\to\alpha$;
\esub
\end{lemma}

\begin{proof} The naive implementations are, in fact, typable.
\end{proof}

We will show that, using the auxiliary EAMs given in Lemma~\ref{lem:welltypedeams}, we can translate any $\EPCF$ term into an EAM. In order to proceed by induction, we first need to define the size of an \EPCF{} term.

\begin{definition}\label{def:termlength}
Let $M$ be an \EPCF{} term and $\sigma$ be an explicit substitution. The sizes $\size{-}$ of $\sigma$, $M$ and $(\sigma,M)$ are defined by mutual induction, e.g.
\[
	\bar{rclcrcl}
	\size{[]} &=& 0,&\qquad&\size{(\sigma,M)} &=& \size{\sigma} + \size{M},\\
	\size{\rho+[x\leftarrow(\rho',N)]} &=& \size{\rho}+\size{(\rho',N)},&&\size{\lam x.\esubst{M}{\sigma}} &=& \size{(\sigma,M)} + 1,\\
	\ear
\]
and the other cases of $\size{M}$ are standard whence they are omitted.
\end{definition}

Intuitively, an \EPCF{} term $M$ having $x_1,\dots,x_n$ as free variables is translated as an EAM $\mach{M}$ loading $n$ arguments as input.

\begin{definition}\label{def:trans} (Translation) Let $M$ be an \EPCF{} term and $\sigma$ be an explicit substitution such that $\FV{M} \subseteq \dom(\sigma) \cup \{\vec x\}$, where $\vec x = x_1,\dots,x_n$. 
The \emph{translation} of the pair $(\sigma, M)$ (w.r.t\ $\vec x$) is a machine denoted $\etrans{\sigma, M}{}{\vec x} \in \cM$, or simply $\trans{M}{}{\vec x}$ when $\sigma$ is empty.\footnote{In other words, we set $\trans{M}{}{\vec x} = \etrans{[], M}{}{\vec x}$.}
The machine $\etrans{\sigma, M}{}{\vec x}$ is defined by induction on $\size{(\sigma, M)}$ as follows:
\[
	\bar{lcll}
	\etrans{\sigma+[y\leftarrow (\tau, N)], M}{\alpha}{\vec x} &=&  
	\appT{\etrans{\sigma, M}{\alpha}{y,\vec x}}{[\Lookup\etrans{\tau, N}{}{}]};\\[0.8ex]
	\trans{x_i}{\alpha}{\vec x} &=& \mProj{n}{i},& 
	\\[0.8ex]	
	\trans{\lambda y.\esubst{M}{\sigma}}{\alpha\to \beta}{\vec x}&= & \etrans{\sigma, M}{\beta}{\vec x, y}, \textrm{ where wlog }y\notin\vec x;
	\\[0.8ex]
	\trans{M\cdot N}{\alpha}{\vec x} &=&  \appT{\mAppn{n}}{[\Lookup\trans{M}{\beta\to\alpha}{\vec x},\Lookup \trans{N}{\beta}{\vec x}]};
	\\[0.8ex]
	\trans{\num k}{\tint}{\vec x} &=& \appT{\mProj{n+1}{1}}{[k]}, \textrm{ where }k\in\nat;
	\\[0.8ex]
	\trans{\pred M}{\tint}{\vec x} &=& \appT{\mPredn{n}}{[\Lookup{\trans{M}{\tint}{\vec x}}]};
	\\[0.8ex]
	\trans{\succ M}{\tint}{\vec x} &=& \appT{\mSuccn{n}}{[\Lookup{\trans{M}{\tint}{\vec x}}]};
	\\[0.8ex]
	\trans{\ifterm LMN}{\alpha}{\vec x} &=& \appT{\mIfZn{n}}{[\Lookup{\trans{L}{\tint}{\vec x}},\Lookup{\trans{M}{\alpha}{\vec x}},\Lookup{\trans{N}{\alpha}{\vec x}}]};
	\\[0.8ex]
	\trans{\fix M}{\alpha}{\vec x} &=& \appT{\mYn{n}}{[\Lookup{\trans{M}{\alpha\to\alpha}{\vec x}}]}.\\
	\ear
\]
\end{definition}

We show the extended abstract machines associated by this translation to some of our running examples.

\begin{example}
1. $\trans{(\lam x.\esubst{\succ (x)}{})\cdot \num 0}{\tint}{} = \appT{\mSuccn{1}}{[\Lookup\mProj{1}{1}, 0]}$.

2. $\trans{(\lambda sn.s(sn))(\lambda x.\succ (x))}{\tint\rightarrow\tint}{} =	\appT{\mAppn{2}}{[\Lookup\mProj{2}{1},\Lookup(\appT{\mAppn{2}}{[\Lookup\mProj{2}{1},\Lookup\mProj{2}{2}]}),\Lookup(\appT{\mSuccn{1}}{[\Lookup\mProj{1}{1}]})]}$.

3. $\bar[t]{l}
\trans{\mathbf{add}}{\tint\to\tint\to\tint}{}=\appT{\mYn{0}}{\Lookup\trans{\lambda f xy. \ifterm y x {(f\cdot (\succ x)\cdot(\pred y))}}{(\tint\rightarrow\tint\rightarrow\tint)\rightarrow\tint\rightarrow\tint\rightarrow\tint}{}}\\
= \appT{\mYn{0}}{\Lookup\trans{\ifterm y x {(f \cdot(\succ x)\cdot(\pred y))}}{\tint}{f,x,y}},\\
= \appT{\mYn{0}}{\Lookup(\appT{\mIfZn{3}}{[\Lookup\mProj{3}{3},\Lookup\mProj{3}{2},\Lookup\trans{f \cdot(\succ x)\cdot(\pred y)}{\tint}{f,x,y}}]}),\\[1ex]
\ear
$\\
where $\trans{f \cdot(\succ x)\cdot(\pred y)}{\tint}{f,x,y} =
 \appT{\mAppn{3}}{[
	\Lookup{
		(\appT{\mAppn{3}}{[
			\Lookup\mProj{3}{1},\Lookup{(
				\appT{\mSuccn{3}}{[\Lookup\mProj{3}{2}]}
			)}
		]})
	}, 
	\Lookup{(\appT{\mPredn{3}}{[\Lookup\mProj{3}{3}]})}
]}$.
\end{example}

\begin{theorem}\label{thm:typabilitytransfers} 
Let $M$ be an \EPCF{} term, $\alpha\in\Types$, $\Gamma = x_1:\beta_1,\dots,x_n:\beta_n$. Then
\[
	\Gamma \vdash^\E M : \alpha \quad \Rightarrow\quad \vdash\trans{M}{}{x_1,\dots,x_n} : \beta_1\to\cdots\to \beta_n\to\alpha.
\]
\begin{proof} By induction on a derivation of  $\Gamma \vdash^\E M : \alpha$. 
As an induction loading, one needs to prove simultaneously that for all explicit substitutions $\sigma$ with $\dom(\sigma)  = \set{x_1,\dots,x_n}$, if $\sigma\models x_1 : \beta_1,\dots,x_n : \beta_n$ then $\vdash \etrans{\sigma(x_i)}{}{} : \beta_i$, for all $i\le n$.
\end{proof}
\end{theorem}


\begin{theorem}\label{thm:equivalence}
Let $M$ be an \EPCF{} term and $V\in\Val$.
Then 
\[
	\substseq{\sigma}{M} \reddd V \Rightarrow \etrans{\sigma, M}{}{} \convh \trans{V}{}{}
\]
\end{theorem}
\begin{proof}
By induction on a derivation of $\substseq{\sigma}{M} \reddd V$.
\end{proof}
\begin{corollary}\label{cor:EPCFprograms} For an \EPCF{} program $M$ of type $\vdash^\E M : \tint$ we have
\[
	\substseq{[\,]}{\,M} \reddd \num n \Rightarrow \trans{M}{}{} \reddh \mach{n}
\]
\end{corollary}

\begin{proof} Assume that $\substseq{[\,]}{M} \reddd \num n$.
By Theorem~\ref{thm:equivalence}, we have $\etrans{[\,],M}{}{} \convh \trans{\num n}{}{}$. Since $\trans{\num n}{}{} \reddh \mach n$ and the numeral machine $\mach{n}$ is in final state we conclude $\trans{M}{\tint}{}\reddh \mach{n}$.
\end{proof}

\subsection{Applying the translation to regular \PCF}
Let us show how to apply our machinery to the usual (call-by-name) \PCF.
Our presentation follows \cite{Ong95}.

\begin{definition}\bsub\item \PCF{} \emph{terms} are defined by the grammar (for $n\ge 0$, $x\in\Var$):
	\[
		P,Q,Q'\ ::=~\,x\mid P\cdot Q \mid \lam x.P
		\mid \mathbf{0} \mid \pred P \mid \succ P\mid \ifterm P{Q}{Q'} \mid \fix P
	\]
	\item A closed \PCF{} term $P$ is called a \PCF{} \emph{program}.
	\item A \PCF{} value $U$ is a term of the form $\lam x.P$ or $\num n$, for some $n\ge 0$. 
	\item Given a \PCF{} term $P$ and a value $U$, we write $P\redd U$ if this judgement can be obtained by applying the rules from Figure~\ref{fig:bigstepPCF}.
	\item The set $\Types$ of \emph{simple types} and \emph{typing contexts} have already been defined in items \eqref{def:simpletypes1} and \eqref{def:simpletypes2} of Definition~\ref{def:simpletypes}, respectively. 
	\item Given a \PCF{} term $P$, a typing context $\Gamma$ and $\alpha\in\Types$, we write $\Gamma\vdash^\PCF P : \alpha$ if this typing judgement is derivable from the rules of Figure~\ref{fig:typesPCF}.
	\esub
\end{definition}

\begin{figure}[t!]
\[
	\bar{ccc}
	\infer[\PCFvalrule]{U\redd U}{ U\in \Val}
	&
	\infer[\predzrule]{\pred P \redd \mathbf{0}}{P \redd \mathbf{0}}	
	&
	\infer[\predrule]{\pred P \redd \num n}{P \redd \num{n+1}}	
	\\[1ex]
	\infer[\ifzzrule]{\ifterm P{Q}{Q'} \redd U_1}{P\redd \num{0} & Q\redd U_1 }
	&\qquad
	\infer[\ifzrule]{\ifterm P{Q}{Q'} \redd U_2}{P\redd \num{n+1} & Q'\redd U_2 }\qquad
	&
	\infer[\succrule]{\succ P \redd \num {n+1}}{P \redd \num{n}}	
	\\[1ex]
	\infer[\fixrule]{\fix P \redd U}{P \cdot (\fix P) \redd U}
	&
	\infer[\betarule]{P\cdot Q \redd U}{P\redd \lam x.P' & P'\subst{x}{Q}\redd U}
	\ear
\]\vspace{-10pt}
\caption{The big-step operational semantics of \PCF.}\label{fig:bigstepPCF}
\end{figure}

\begin{figure}[t]
\[
	\bar{ccccccc}
	\infer{\Gamma,x:\alpha\vdash x:\alpha}{}
	&&
	\infer{\Gamma\vdash\lambda x.P:\alpha\to \beta}{\Gamma,x:\alpha\vdash P:\beta}
	&&
	\infer{\Gamma\vdash\num 0:\tint}{}&&\infer{\Gamma\vdash PQ:\beta}{
  \Gamma\vdash P:\alpha\to \beta
  &
  \Gamma\vdash Q:\alpha
}\\[1ex]
  \infer{\Gamma\vdash\pred~P:\tint}{\Gamma\vdash P:\tint}.
  &&
  \infer{\Gamma\vdash \fix P:\alpha}{\Gamma\vdash P:\alpha\to \alpha}
  &&
\infer{\Gamma\vdash\succ~P:\tint}{\Gamma\vdash P:\tint}. 
  &&\infer{\Gamma\vdash \ifterm P{Q}{Q'}:\alpha}{
	\Gamma\vdash P:\tint
	&
	\Gamma\vdash Q:\alpha
	&
	\Gamma\vdash Q':\alpha
}\\
	\ear
\]
\caption{The type inference rules of \PCF.}\label{fig:typesPCF}
\end{figure}

Recall that any \PCF{} program $P$ can be seen as an $\EPCF$ term, thanks to the notation $\lam x.N := \lam x.\esubst{N}{}$.
However, the hypotheses $\vdash^\PCF P : \tint$ and $P\redd\num n$ are \emph{a priori} not sufficient for applying Corollary~\ref{cor:EPCFprograms}, since one needs to show that also the corresponding \EPCF{} judgments $\vdash^\E P : \tint$ and $P\reddd\num n$ hold.
The former is established by the following lemma.
\begin{lemma}\label{lem:typeequivalence} 
Let $M$ be a \PCF{} term, $\alpha\in\Types$ and $\Gamma$ be a context. Then
\[
	\Gamma\vdash^{\PCF} M : \alpha\qquad \Rightarrow\qquad \Gamma\vdash^\E M : \alpha
\]
\end{lemma}
\begin{proof} By a straightforward induction on a derivation of $\Gamma\vdash^{\PCF} M$.
\end{proof}

An \EPCF{} term is easily translated into \PCF{} by performing all its explicit substitutions. 
The converse is trickier as the representation is not unique: for every \PCF{} term $P$ there are several decompositions $P = P'[Q_1/x_1,\dots,Q_n/x_n]$. Recall that the size $\size{(\sigma,M)}$ has been defined in Definition~\ref{def:termlength}.

\begin{definition}\label{def:PCFsubstitution}
Let $M$ be an \EPCF{} term and $\sigma$ be an explicit substitution. Define a \PCF{} term $(\sigma, M)^*$ by induction on $\size{(\sigma,M)}$ as follows:
\[
	\bar{rcl}
	(\sigma, x)^*&=&\begin{cases}
		\sigma(x)^*,&\textrm{if }x\in\dom(\sigma),\\
		x,&\textrm{otherwise},
	\end{cases}\\
	(\sigma, \lambda x.\esubst{M}{\rho})^* &=& \lambda x.(\sigma+\rho, M)^*,	\\
	(\sigma, M\cdot N)^* &=& (\sigma, M)^* \cdot (\sigma, N)^*,\\	
	(\sigma, \fix M)^* &=& \fix ((\sigma, M)^*),\\
	(\sigma, \num 0)^* &=& \num 0,\\
	(\sigma, \pred M)^* &=& \pred (\sigma, M)^*,\\
	(\sigma, \succ M)^* &=& \succ (\sigma, M)^*,\\
	(\sigma, \ifterm L {M}{N})^* &=& \ifterm{ (\sigma, L)^*}{(\sigma, M)^*}{(\sigma, N)^*}.\\	
	\ear
\]
For a \PCF{} term $P$, define $P^\dagger = \set{(\sigma, M) \st (\sigma, M)^* = P}$.
\end{definition}

To show the equivalence between \PCF{} and \EPCF, we need yet another auxiliary lemma.
\begin{lemma}[Substitution Lemma]\label{app:lem:substbizarre}\ 
\bsub\item\label{app:lem:substbizarre1}
Let $M,N$ be \EPCF{} terms, $\sigma, \rho$ be explicit substitutions and $x$ be a variable.
\[
	(\sigma+[x \leftarrow(\rho, N)], M)^* =(\sigma,M)^*\subst{x}{(\rho, N)^*}
\]
\item\label{app:lem:substbizarre2} 
Let $P,Q$ be \PCF{} terms with $\FV{P}\subseteq\set{x}$ and $Q$ closed. 
For all \EPCF{} terms $M,N$ and explicit substitutions $\sigma,\rho$, we have:
\[
	(\sigma,M)\in P^\dagger \land (\rho,N)\in Q^\dagger\Rightarrow
	(\sigma + [x\leftarrow (\rho,N)],M)\in(P\subst{x}{Q})^\dagger
\]
\esub
\end{lemma}
We rely on the freshness hypothesis on the variables in $\dom(\sigma)$.
\begin{proof} (i)
By structural induction on $M$. 
\begin{description}
   \item[Case $M = y$, with $y\neq x$:] There are two subcases. 
   \begin{itemize}
   \item
   If $y\in \dom(\sigma)$, then
$
(\sigma+[x \leftarrow(\rho,N)], y)^* = \sigma(y) = \sigma(y)\subst{x}{(\rho, N)^*}= (\sigma, y)^*\subst{x}{(\rho, N)^*}$, since $x\notin\FV{\sigma(y)}$;
   \item
	if $y\notin \dom(\sigma)$, then $
	(\sigma+[x \leftarrow(\rho,N)], y)^* = y 
	= y\subst{x}{(\rho, N)^*}
	= (\sigma, y)^*\subst{x}{(\rho, N)^*}$.
	\end{itemize}
\item[Case $M = x$:] Then 
$
(\sigma+[x \leftarrow(\rho, N)], x)^*  = (\rho, N)^*
=x\subst{x}{(\rho, N)^*}
=(\sigma, x)^*\subst{x}{(\rho, N)^*}.
$

\item[Case $M = \lambda y.M'$:] Wlog, we may assume $y\neq x$. We have
$
(\sigma+[x \leftarrow(\rho, N)], \lambda y.M')^*= \lambda y.(\sigma+[x \leftarrow(\rho, N)], M')^*
= \lambda y.((\sigma,M')^*\subst{x}{(\rho, N)^*})
=(\lambda y.(\sigma,M')^*)\subst{x}{(\rho, N)^*} 
=((\sigma,\lambda y.M')^*)\subst{x}{(\rho, N)^*}.
$
\end{description}
All other cases derive straightforwardly from the IH.

(ii) By an easy induction on $P$, using \eqref{app:lem:substbizarre1}.
\end{proof}

\begin{theorem}\label{thm:delayed_equivalence} 
The big-step weak reduction of \EPCF{} is equivalent to the usual big-step operational semantics of \PCF.  Formally:
\bsub
\item\label{thm:delayed_equivalence1} Given an \EPCF{} program $M$, a value $V$ and an explicit substitution $\sigma$, we have:
\[
	\substseq{\sigma}{ M} \reddd V \imp
	(\sigma, M)^* \redd ([\,], V)^*
\]
\item\label{thm:delayed_equivalence2} Given a \PCF{} program $P$ and \PCF{} value $U$. If $P \redd U$ then
\[
	  \forall (\sigma, M) \in P^\dagger,\, \exists V \in \Val\,.\,(\ \substseq{\sigma}{M}\reddd V \textrm{ and } ([\,], V)\in U^\dagger \ )
\]
\esub
\begin{proof}(Proof sketch) For the full proof, we refer to the technical Appendix~\ref{app:tech}.

(i) Proceed by induction on a derivation of $\substseq{\sigma}{ M} \reddd V$, using  Lemma~\ref{app:lem:substbizarre}\eqref{app:lem:substbizarre1} in the $(\beta_v)$-case.

(ii) By induction on the lexicographically ordered pairs, whose first component is the length of a derivation of $P \redd U$ and second component is $\size{(\sigma,M)}$, using  Lemma~\ref{app:lem:substbizarre}\eqref{app:lem:substbizarre2} in the $(\beta_v)$-case.
\end{proof}
\end{theorem}

As promised, we now draw conclusions for the regular \PCF.
As customary in \PCF, we are interested on the properties of closed terms having ground type.

\begin{theorem}\label{thm:main}
For a \PCF{} program $P$ of type $\tint$, $P \redd \num n$ entails $\trans{P}{\tint}{}\reddh \mach{n}$.
\end{theorem}
\begin{proof} Note that $P$ is also an $\EPCF$ term such that $([],P)\in P^\dagger$, and that $\vdash^\E P : \tint$ by Lemma~\ref{lem:typeequivalence}.
Thus $[\,]\triangleright\, P\reddd \num n$ by Theorem~\ref{thm:delayed_equivalence}\eqref{thm:delayed_equivalence2}.
Conclude by Corollary~\ref{cor:EPCFprograms}.
\end{proof}


\bibliographystyle{entics}

\bibliography{include/biblio}
\appendix
\section{Technical Appendix}\label{app:tech}

This technical appendix is devoted to provide the proofs that have been partially given, or completely omitted, in the body of the paper.
As an abbreviation, we write IH for ``induction hypothesis''.

\subsection{Proofs of Section~\ref{sec:EPCF}}

\begin{proof}[Proof of Lemma~\ref{lem:epcftyping}] Items \eqref{lem:epcftyping1} and \eqref{lem:epcftyping2} are straightforward. We prove \eqref{lem:epcftyping4}.

Given an \EPCF{} term $M$, an \EPCF{} value $V$, an explicit substitution $\sigma$, a type $\alpha$, and a context $\Gamma$ such that $\substseq{\sigma}M \reddd V, \sigma \models \Gamma, \Gamma \vdash M : \alpha$, we prove by induction on a derivation of $\substseq{\sigma}M \reddd V$ that $\vdash V : \alpha$.

\begin{description}
\item [Case $\valrule$:] In this case $M = V = \num n$, for $ \num n \in \nat$, and $\alpha = \tint$. By the typing rules of \EPCF{}, $\vdash \num n : \tint$.

\item [Case $\funrule$:] In this case $M = \lam x.\esubst{M'}{\rho}, V = \lam x.\esubst{M'}{\sigma + \rho}, \alpha = \beta\rightarrow\gamma$. As $\Gamma \vdash \lam x.\esubst{M'}{\rho} : \beta\rightarrow\gamma$, $\exists! \Delta$ such that $\rho \models \Delta, \Gamma, \Delta, x:\beta \vdash M' : \gamma$. As $\sigma \models \Gamma$ and $\rho \models \Delta$, $\sigma + \rho \models \Gamma, \Delta$. Thus by the typing rules of \EPCF{}, $\vdash \lam x.\esubst{M'}{\sigma + \rho} : \beta\rightarrow\gamma$.

\item [Case $\varrule$:] In this case $M = x, \sigma(x) = (\rho, N), \Gamma = \Gamma', x:\alpha$. By the operational semantics of \EPCF{}, $\substseq{\rho}{N}\reddd V$, and by the type system of \EPCF{} as $\sigma \models \Gamma', x:\alpha$, $[x\leftarrow (\rho, N)] \models x: \alpha$ and thus $\rho \models \Delta, \Delta \vdash N : \alpha$. By IH, $\vdash V : \alpha$.

\item [Case $\betarule$:] In this case $M = N \cdot L$. By the operational semantics of \EPCF{}, $\substseq{\sigma}{N}\reddd \lam x.\esubst{N'}{\rho}$ and $\substseq{\rho + [x \leftarrow (\sigma, L)]}{N'} \reddd V$. By the type system of \EPCF{}, $\Gamma \vdash N : \beta\rightarrow\alpha$, $\Gamma \vdash L : \beta$. By IH, $\vdash \lam x.\esubst{N'}{\rho} : \beta\rightarrow\alpha$, and then by the type system of \EPCF{} $\rho \models \Delta, \Delta, x:\beta \vdash N' : \alpha$. By the type system we also have $[x\leftarrow(\sigma, L)]\models x:\beta$, so $\rho + [x\leftarrow(\sigma, L)]\models\Delta,x:\beta$, and thus by IH we conclude $\vdash V:\alpha$.
\end{description}
All other cases derive straightforwardly from applying the rules of the type system and the IH.
\end{proof}

\subsection{Proofs of Section~\ref{sec:trans}}

\begin{proof}[Proof of Theorem~\ref{thm:typabilitytransfers}] We prove the following statements by mutual induction and call the respective inductive hypotheses IH1 and IH2.

\begin{enumerate}[(i)]
\item\label{thm:typabilitytransfers1} 
Let $M$ be an \EPCF{} term, $\Gamma = x_1:\beta_1,\dots,x_n:\beta_n$ and $\alpha\in\Types$. Then
$
	\Gamma \vdash M : \alpha \quad \Rightarrow\quad \vdash\trans{M}{}{x_1,\dots,x_n} : \beta_1\to\cdots\to\beta_n\to\alpha.
$

\item\label{thm:typabilitytransfers2} 
For all $\Gamma = x_1:\beta_1,\dots,x_n:\beta_n$ and $\sigma$, with $\dom(\sigma)  = \set{x_1,\dots,x_n}$, we have 
$
	\sigma\models \Gamma \quad \Rightarrow \quad \vdash \etrans{\sigma(x_1)}{}{} : \beta_1 ,\dots,\vdash \etrans{\sigma(x_n)}{}{} : \beta_n.
$
\end{enumerate}
We start with the cases concerning \ref{thm:typabilitytransfers1}.
\begin{description}
\item[Case $(\mathrm{ax})$:] Then, $ x_1:\beta_1,\dots,x_n : \beta_n \vdash x_i : \beta_i$, and $\trans{x_i}{}{x_1,\dots,x_n} = \mProj{n}{i}.$
By Lemma~\ref{lem:welltypedeams}\eqref{lem:welltypedeams1} we conclude $\vdash \mProj{n}{i} : \beta_1\to\dots\to\beta_n\to\beta_i$.\medskip
\item[Case $(0)$:] In this case, we have $\Gamma \vdash \mathbf{0} : \tint$ and 
$
	\trans{\mathbf{0}}{}{x_1,\dots,x_n} = \appT{\mProj{n+1}{1}}{[0]}.
$
By Lemma~\ref{lem:welltypedeams}\eqref{lem:welltypedeams1} $\vdash \appT{\mProj{n+1}{1}} : \tint\to\beta_1\to\dots\to\beta_n\to\tint$ and, by Figure~\ref{fig:eamstyping} $\vdash\Lookinv{0}:\tint$.
By Proposition~\ref{prop:typing}\eqref{prop:typing1}, $\vdash\appT{\mProj{n+1}{1}}{[{0}]}: \beta_1\to\dots\to\beta_n\to\tint$.\medskip

\item[Case $(\mathrm{Y})$:] In this case $\Gamma \vdash \fix M' : \alpha$ and $
	\trans{\fix M'}{}{\vec x} = \appT{\mYn{n}}{[\Lookup{\trans{M'}{}{\vec x}}]}$. 
By $(\mathrm{fix}_n)$ in Figure~\ref{fig:eamstyping}, we have $\vdash \mYn{n} : (\vec \beta\to \alpha\to\alpha)\to\vec\beta\to\alpha$. 
From the hypothesis IH1, we obtain $\vdash \trans{M'}{}{x_1,\dots,x_n} : \vec\beta \to \alpha\to\alpha$. 
By Proposition~\ref{prop:typing}\eqref{prop:typing1}, we conclude that $\vdash\appT{\mYn{n}}{[\Lookup{\trans{M'}{}{x_1,\dots,x_n}}]} : \beta_1\to\dots\to\beta_n\to\alpha$.
\medskip
\item[Case $(+)$:] In this case $\Gamma \vdash \succ {M'} : \tint$ since $\Gamma \vdash {M'} : \tint$. 
By definition, $\trans{\succ M'}{}{x_1,\dots,x_n} = \appT{\mSuccn{n}}{[\Lookup\trans{M'}{}{x_1,\dots,x_n}]}$. 
By Lemma~\ref{lem:welltypedeams}\eqref{lem:welltypedeams4}, $\vdash \mSuccn{n} : (\vec\beta \to \tint)\to\vec\beta\to \tint$. 
From IH1, we get $\vdash \trans{M'}{}{x_1,\dots,x_n} : \vec\beta\to \tint$, and thus by Proposition~\ref{prop:typing}\eqref{prop:typing1}, we conclude 
$
	\vdash\appT{\mSuccn{n}}{[\Lookup{\trans{M'}{}{x_1,\dots,x_n}}]} : \vec\beta\to\tint.
$

\item[Case $(-)$:] Analogous, applying Lemma~\ref{lem:welltypedeams}\eqref{lem:welltypedeams4}.\medskip

\item[Case $(\mathrm{ifz})$:] Assume $\Gamma\vdash \ifterm{L}{N_1}{N_2} : \alpha$ since $\Gamma \vdash L : \tint$ and $\Gamma \vdash N_i : \alpha$, for $i\in\set{1,2}$. 
By definition of the translation, we have 
$
	\trans{\ifterm L{N_2}{N_2}}{}{\vec x} = \appT{\mIfZn{n}}{[\Lookup\trans{L}{}{\vec x},\Lookup\trans{N_1}{}{\vec x},\Lookup\trans{N_2}{}{\vec x}]}.
$ 
By applying Lemma~\ref{lem:welltypedeams}\eqref{lem:welltypedeams5}, we obtain $\vdash \mIfZn{n} : (\vec\beta\to \tint)\to(\vec\beta\to \alpha)\to(\vec\beta\to \alpha)\to\vec\beta\to\alpha$. 
By the hypothesis IH1, we get $\vdash \trans{L}{}{\vec x} : \vec\beta\to\tint$ and $\vdash \trans{N_i}{}{\vec x} : \vec\beta \to \alpha$ for $i\in\set{1,2}$, and thus by Proposition~\ref{prop:typing}\eqref{prop:typing1}, we conclude $\vdash\appT{\mIfZn{n}}{[\Lookup\trans{L}{}{\vec x},\Lookup\trans{N_1}{}{\vec x},\Lookup\trans{N_2}{}{\vec x}]} : \vec\beta\to\alpha$.\medskip

\item[Case $(\to_\mathrm{I})$:]  Assume that $\Gamma\vdash \lam z.\esubst{M'}{\sigma} : \alpha_1\to\alpha_2$, for $\alpha = \alpha_1\to\alpha_2$, because there is $\Delta = y_1:\delta_1,\dots,y_m : \delta_m$ such that $\sigma\models \Delta$ and $\Gamma,\Delta,z:\alpha_1\vdash M' : \alpha_2$. 
Then $\sigma\models \Delta$ entails $\sigma = [y_1\leftarrow(\rho_1,N_1),\dots,y_m\leftarrow(\rho_m,N_m)]$ for appropriate $\vec \rho,\vec N$.
By definition, we have
$
	\trans{\lam z.\esubst{M'}{\sigma}}{}{x_1,\dots ,x_n}
	 = \etrans{\sigma, M'}{}{\vec x, z}
	 = \appT{\trans{M'}{}{y_1,\dots,y_m,\vec x, z}}{[\Lookup\etrans{\rho_1,N_1}{}{},\dots,\Lookup\etrans{\rho_m,N_m}{}{}]}. 
$
By applying IH1, we obtain $\vdash\trans{M'}{}{\vec y,\vec x, z}: \vec\delta\to\vec\beta\to\alpha_1\to\alpha_2$. 
From IH2, we get $\vdash \etrans{\sigma_1,N_1}{}{}:\delta_1 \dots \vdash \etrans{\sigma_m,N_m}{}{} : \delta_m$.
Finally, by Proposition~\ref{prop:typing}\eqref{prop:typing1}, we derive $\vdash \appT{\trans{M'}{}{\vec y,\vec x, z}}{[\Lookup\etrans{\sigma_1,N_1}{}{},\dots,\Lookup\etrans{\sigma_m,N_m}{}{}]} : \vec \beta\to\alpha_1\to\alpha_2$.
\item[Case $(\to_\mathrm{E})$:] In this case $\Gamma\vdash M_1\cdot M_2 : \alpha$ since, for some $\delta\in\Types$, $\Gamma\vdash M_1 : \delta\to\alpha$ and $\Gamma \vdash M_2 : \delta$.
By definition, 
$
	\trans{M_1\cdot M_2}{}{\vec x} = \appT{\mAppn{n}}{[\Lookup\trans{M_1}{}{\vec x}, \Lookup\trans{M_2}{}{\vec x}]}. 
$
By Lemma~\ref{lem:welltypedeams}\eqref{lem:welltypedeams2}, we get $\vdash \mAppn{n} : (\vec\beta\to\delta\to\alpha)\to(\vec\beta\to\delta)\to\vec\beta\to\alpha$. By IH1, we obtain $\vdash \trans{M_1}{}{\vec x} : \vec\beta\to\delta\to\alpha$ and $\vdash \trans{M_2}{}{\vec x} : \vec\beta\to\delta$. 
Conclude, by Proposition~\ref{prop:typing}\eqref{prop:typing1}, that $\vdash\appT{\mAppn{n}}{[\Lookup\trans{M_1}{}{\vec x}, \Lookup\trans{M_2}{}{\vec x}]} : \vec\beta\to\alpha$.
\end{description}
We now consider the cases concerning \ref{thm:typabilitytransfers2}.
\begin{description}
\item[Case $(\sigma_0)$]: In this case $[\,] \models \emptyset$, so we have nothing to prove.\medskip
\item[Case $(\sigma)$:] In this case $\Gamma = \Gamma',x_n:\beta_n$ and $\sigma = \sigma' + [x_n\leftarrow(\rho,N)]\models \Gamma',x_n:\beta_n$, 
because $\sigma'\models \Gamma'$, $\rho\models \Delta$ and $\Delta\vdash N :\beta_n$, for some $\Delta = y_1 : \delta_1,\dots,y_m : \delta_m$.
By IH2 on $\sigma'\models \Gamma'$, we get $\etrans{\sigma(x_i)}{}{} = \etrans{\sigma'(x_i)}{}{} : \beta_i$, for all $i\in\set{1,\dots,n-1}$. 
We show $\etrans{(\rho,N)}{}{} : \beta_n$.
By IH1, we get $\vdash \trans{N}{}{y_1,\dots,y_m} : \delta_1\to\cdots\to\delta_m\to\beta_n$.
By applying IH2 on $\rho\models \Delta$, we get $\etrans{\rho(y_j)}{}{} : \delta_j$, for all $j \in \{1,\dots, m\}$.
Since $\etrans{\rho, N}{}{} = \appT{\trans{N}{}{\vec y}}{[\Lookup\etrans{\rho(y_1)}{}{},\dots,\Lookup\etrans{\rho(y_m)}{}{}]}$, we conclude $\vdash \etrans{\rho, N}{}{} : \beta_n$, by Proposition~\ref{prop:typing}\eqref{prop:typing1}.
\end{description}
\end{proof}

\begin{proof}[Proof of Theorem~\ref{thm:equivalence}] We prove
$
	\substseq{\sigma}{M} \reddd V \Rightarrow \etrans{\sigma, M}{}{} \convh \trans{V}{}{}
$
by induction on a derivation of $\substseq{\sigma}{M} \reddd V$. We let $\dom(\sigma) = \set{x_1,\dots,x_n}$ and sometimes use the convenient notation 
$
	\Lookup\etrans{\sigma(\vec x)}{}{} = \Lookup\etrans{\sigma(x_1)}{}{},\dots, \Lookup\etrans{\sigma(x_n)}{}{}. 
$
\begin{description}
\item[Case $\valrule$:] Then $M =V=\num k$, for some $k\ge 0$. 
Recall that we assume $k = \Lookup\mach{k}$. Then using Lemma~\ref{lem:existenceofeams}\eqref{lem:existenceofeams1}, 
$
	\etrans{\sigma,\num k}{}{}=\appT{\trans{\num k}{}{x_1,\dots,x_n}}{[\Lookup\etrans{\sigma(\vec x)}{}{}]}=\appT{\mProj{n+1}{1}}{[k,\Lookup\etrans{\sigma(\vec x)}{}{}]}
	\reddh \Lookinv{k} {}_\mach{c}\!\!\twoheadleftarrow \appT{\mProj{1}{1}}{[k]} = \trans{\num k}{}{}
$
\item[Case 2: $\funrule$] We have $M = \lam z.\esubst{M'}{\rho}$ and $ V = \lam x.\esubst{M'}{\sigma+\rho}$, with $\dom(\sigma)\cap\dom(\rho) = \emptyset$.
Say, $\dom(\rho) = \set{y_1,\dots,y_m}$. Then
\begin{align*}
\etrans{\sigma, \lam z.\esubst{M'}{\rho}}{}{}&=\appT{\trans{\lam z.\esubst{M'}{\rho}}{}{x_1,\dots,x_n}}{[\Lookup\etrans{\sigma(\vec x)}{}{}]}=\appT{\etrans{\rho, M'}{}{\vec x,z}}{[\Lookup\etrans{\sigma(\vec x)}{}{}]}\\
&=\appT{\trans{M'}{}{\vec y,\vec x,z}}{[\Lookup\etrans{\rho(\vec y)}{}{},\Lookup\etrans{\sigma(\vec x)}{}{}]}=\etrans{\sigma + \rho, M'}{}{z}=\trans{\lam z.\esubst{M'}{\sigma+\rho}}{}{}
\end{align*}
The case follows by reflexivity of $\convh$.\medskip
\item[Case $\varrule$:] We have $M = x_i$, $\sigma(x_i) = (\rho,N)$ and $\substseq{\rho}{N\reddd V}$.
Then, using Lemma~\ref{lem:existenceofeams}\eqref{lem:existenceofeams1}, we have
$
\etrans{\sigma, x_i}{}{} = \appT{\trans{x_i}{}{x_1,\dots,x_n}}{[\Lookup\etrans{\sigma(\vec x)}{}{}]}
= \appT{\mProj{n}{i}}{[\Lookup\etrans{\sigma(\vec x)}{}{}]}
\reddh \etrans{\sigma(x_i)}{}{} = \etrans{\rho, N}{}{}
$
We conclude since, by IH, we have $\etrans{\rho, N}{}{}\convh \trans{V}{}{}$.\medskip
\item[Case $\betarule$:] $M = M_1\cdot M_2$, $\substseq{\sigma}{M_1} \reddd \lam z.\esubst{M_1'}{\rho}$ and $\substseq{\rho + [z\leftarrow(\sigma, M_2)]}{M_1'}\reddd V$.
Easy calculations give:
\begin{align*}
&\etrans{\sigma, M_1 \cdot M_2}{}{}\\
=& \appT{\trans{M_1\cdot M_2}{}{x_1,\dots,x_n}}{[\Lookup\etrans{\sigma(x_1)}{}{},\dots, \Lookup\etrans{\sigma(x_n)}{}{}]}\\
=&\appT{\mAppn{n}}{[\trans{M_1}{}{\vec x},\Lookup\trans{M_2}{}{\vec x},\Lookup\etrans{\sigma(x_1)}{}{},\dots, \Lookup\etrans{\sigma(x_n)}{}{}]}\\
\reddh& \appT{\trans{M_1}{}{\vec x}}{[\Lookup\etrans{\sigma(\vec x)}{}{},\Lookup{(\appT{\trans{M_2}{}{\vec x}}{[\Lookup\etrans{\sigma(\vec x)}{}{}]})}]},&\textrm{by Lemma~\ref{lem:existenceofeams}\eqref{lem:existenceofeams2}},\\
=& \appT{\etrans{\sigma, M_1}{}{}}{[\Lookup{\etrans{\sigma, M_2}{}{}}]}
\end{align*}
By IH on $\substseq{\sigma}{M_1} \reddd \lam z.\esubst{M_1'}{\rho}$, we have $\etrans{\sigma, M_1}{}{} \convh \trans{\lam z.\esubst{M_1'}{\rho}}{}{} = \etrans{\rho, M_1'}{}{z}$. Calling $\dom(\rho) = \set{y_1,\dots,y_m}$, we get
\begin{align*}
\appT{\etrans{\sigma, M_1}{}{}}{[\Lookup{\etrans{\sigma, M_2}{}{}}]}
&\convh \appT{\etrans{\rho, M_1'}{}{z}}{[\Lookup{\etrans{\sigma, M_2}{}{}}]},& \textrm{by def.\ of $\convh$ and Lemma~\ref{lem:reduction}},\\
&=\appT{\trans{M_1'}{}{y_1,\dots,y_m, z}}{[\Lookup\etrans{\rho(\vec y)}{}{}, \Lookup{\etrans{\sigma, M_2}{}{}}]}\\
&=\etrans{\rho + [z\leftarrow(\sigma, M_2)], M_1'}{}{}
\end{align*}
By IH conclude $\etrans{\rho + [z\leftarrow(\sigma, M_2)], M_1'}{}{} \convh \trans{V}{}{}$. 
\item[Case $\predrule$:] $M = \pred M'$, $V = \num n$ and $\substseq{\sigma}{M' \reddd\num{n+1}}$, for some $n\ge 0$.
By IH we get $\etrans{\sigma, M'}{}{}\convh \trans{\num {n+1}}{}{} = \Lookinv{n+1}$ and, since the $(n+1)$-th numeral machine is in final state, we derive $\etrans{\sigma, M'}{}{}\reddh \Lookinv{n+1}$. 
Conclude as follows:
\[
\begin{array}{rll}
\etrans{\sigma, \pred M'}{}{}
=& \appT{\trans{\pred M'}{}{x_1,\dots,x_n}}{[\Lookup\etrans{\sigma(x_1)}{}{},\dots, \Lookup\etrans{\sigma(x_n)}{}{}]}\\
=& \appT{\mPredn{n}}{[\Lookup\trans{M'}{}{x_1,\dots,x_n},\Lookup\etrans{\sigma(x_1)}{}{},\dots, \Lookup\etrans{\sigma(x_n)}{}{}]}\\
\reddh& \tuple{R_0 = \Lookup(\appT{\trans{M'}{}{\vec x}}{[\Lookup\etrans{\sigma(\vec x)}{}{}]}),\vec R,\Pred00;\Call 0,[]},&\textrm{by Lemma~\ref{lem:existenceofeams}\eqref{lem:existenceofeams3},}\\
=& \tuple{R_0 = \Lookup\etrans{\sigma, M'}{}{},\vec R,;\Pred00;\Call 0,[]}\\
\reddh& \tuple{R_0 = n+1,\vec R,\Pred00;\Call 0,[]}\reddh \Lookinv{n}
\end{array}
\]	
\item[Case $\predzrule$:] Analogous to the previous case, using the fact that $\ins{Pred}(0)$ is 0.\medskip

\item[Case 7: $\succrule$] $M = \succ M'$, $V = \num {n+1}$ and $\substseq{\sigma}{M' \reddd\num n}$, for some $n\ge 0$.
By IH we get $\etrans{\sigma, M'}{}{}\convh \trans{\num {n}}{}{} = \Lookinv{n}$ and, since the $n$-th numeral machine is in final state, we derive $\etrans{\sigma, M'}{}{}\reddh \Lookinv{n}$. 
Conclude as follows:
\[
\begin{array}{rlll}
\etrans{\sigma, \succ M'}{}{}&=& \appT{\trans{\succ M'}{}{x_1,\dots,x_n}}{[\Lookup\etrans{\sigma(x_1)}{}{},\dots, \Lookup\etrans{\sigma(x_n)}{}{}]}\\
&=& \appT{\mSuccn{n}}{[\Lookup\trans{M'}{}{x_1,\dots,x_n},\Lookup\etrans{\sigma(x_1)}{}{},\dots, \Lookup\etrans{\sigma(x_n)}{}{}]}\\
&\reddh& \tuple{R_0 = \Lookup(\appT{\trans{M'}{}{\vec x}}{[\Lookup\etrans{\sigma(\vec x)}{}{}]}),\vec R,\Succ00;\Call 0,[]},&\textrm{by Lemma~\ref{lem:existenceofeams}\eqref{lem:existenceofeams3},}\\
&=& \tuple{R_0 = \Lookup\etrans{\sigma, M'}{}{},\vec R,;\Succ00;\Call 0,[]}\\
&\reddh& \tuple{R_0 = n,\vec R,\Succ00;\Call 0,[]}\reddh& \Lookinv{n + 1}\\
\end{array}
\]	

\item[Case 8: $\ifzrule$] $M = \ifterm L {N_1} {N_2}$, $\substseq{\sigma}{L}\reddd \num {n+1}$, $\substseq{\sigma}{N_2} \reddd V$ for some $n \ge 0$. By IH we get $\etrans{\sigma, L}{}{}\convh\trans{\num {n+1}}{}{}$, and since the $(n+1)$-th numeral machine is in final state, we derive $\etrans{\sigma, L}{}{}\reddh \Lookinv{n+1}$. Then
\[
\begin{array}{lcl}
\etrans{\sigma, \ifterm L {N_1}{N_2}}{}{}&=& \appT{\trans{\ifterm L{N_1}{N_2}}{}{x_1,\dots,x_n}}{[\Lookup\etrans{\sigma(x_1)}{}{},\dots, \Lookup\etrans{\sigma(x_n)}{}{}]}\\
&=& \appT{\mIfZn{n}}{[\Lookup\trans{L}{}{x_1,\dots,x_n},\Lookup\trans{N_1}{}{x_1,\dots,x_n},\Lookup\trans{N_2}{}{x_1,\dots,x_n}, \Lookup\etrans{\sigma(x_1)}{}{},\dots, \Lookup\etrans{\sigma(x_n)}{}{}]}\\
&\reddh& \Tuple{R_0 = \Lookup\etrans{\sigma, L}{}{},R_1 = \Lookup\etrans{\sigma,N_1}{}{},R_2 = \Lookup\etrans{\sigma,N_2}{}{},\vec R, \Ifz 0120;\Call 0,[]},\\
&&\textrm{by Lemma~\ref{lem:existenceofeams}(\ref{lem:existenceofeams5}),}\\
&\reddh& \Tuple{R_0 = {n+1},R_1 = \Lookup\etrans{\sigma,N_1}{}{},R_2 = \Lookup\etrans{\sigma,N_2}{}{},\vec R, \Ifz 0120;\Call 0,[]}\\
&\reddh& \etrans{\sigma,N_2}{}{}
\end{array}
\]
We conclude since, by IH, $\etrans{\sigma, N_2}{}{}\convh\trans{V}{}{}$.

\item[Case 9: $\ifzzrule$] $M = \ifterm L {N_1} {N_2}$, $\substseq{\sigma}{L}\reddd \num {0}$, $\substseq{\sigma}{N_1} \reddd V$. By IH we get $\etrans{\sigma, L}{}{}\convh\trans{\num {0}}{}{}$, and since the $(n+1)$-th numeral machine is in final state, we derive $\etrans{\sigma, L}{}{}\reddh \Lookinv{0}$. Then
\[
\bar{rl}
\etrans{\sigma, \ifterm L {N_1}{N_2}}{}{}
=& \appT{\trans{\ifterm L{N_1}{N_2}}{}{x_1,\dots,x_n}}{[\Lookup\etrans{\sigma(x_1)}{}{},\dots, \Lookup\etrans{\sigma(x_n)}{}{}]}\\
=& \appT{\mIfZn{n}}{[\Lookup\trans{L}{}{x_1,\dots,x_n},\Lookup\trans{N_1}{}{x_1,\dots,x_n},\Lookup\trans{N_2}{}{x_1,\dots,x_n}, \Lookup\etrans{\sigma(x_1)}{}{},\dots, \Lookup\etrans{\sigma(x_n)}{}{}]}\\
\reddh& \Tuple{R_0 = \Lookup\etrans{\sigma, L}{}{},R_1 = \Lookup\etrans{\sigma,N_1}{}{},R_2 = \Lookup\etrans{\sigma,N_2}{}{},\vec R, \Ifz 0120;\Call 0,[]},\\
&\textrm{by Lemma~\ref{lem:existenceofeams}(\ref{lem:existenceofeams5}),}\\
\reddh& \Tuple{R_0 = {0},R_1 = \Lookup\etrans{\sigma,N_1}{}{},R_2 = \Lookup\etrans{\sigma,N_2}{}{},\vec R, \Ifz 0120;\Call 0,[]}\\
\reddh& \etrans{\sigma,N_1}{}{}
\ear
\]
We conclude since, by IH, $\etrans{\sigma, N_1}{}{}\convh\trans{V}{}{}$.

\item[Case $\fixrule$:] Then $M = \fix M'$ and $\substseq{\sigma}{M'\cdot(\fix M')}$. Then
\[
\bar{lcll}
\etrans{\sigma, \fix M'}{}{}&=& \appT{\trans{\fix M'}{}{x_1,\dots,x_n}}{[\Lookup\etrans{\sigma(x_1)}{}{},\dots,\Lookup\etrans{\sigma(x_n)}{}{}]}\\
&=& \appT{\mYn{n}}{[\Lookup\trans{M'}{}{x_1,\dots,x_n},\Lookup\etrans{\sigma(x_1)}{}{},\dots,\Lookup\etrans{\sigma(x_n)}{}{}]}\\
&\reddh& \appT{\trans{M'}{}{\vec x}}{[\Lookup\etrans{\sigma(\vec x)}{}{},\Lookup(\appT{\mYn{n}}{[\Lookup\trans{M'}{}{\vec x},\Lookup\etrans{\sigma(\vec x)}{}{}]})]}\\
&=& \appT{\trans{M'}{}{\vec x}}{[\Lookup\etrans{\sigma(\vec x)}{}{},
\Lookup(
\appT{\trans{\fix M'}{}{\vec x}}{[\Lookup\etrans{\sigma(\vec x)}{}{}]}
)]},&\textrm{by Lemma~\ref{lem:existenceofeams}\eqref{lem:existenceofeams2}},\\
&{}_\mach{c}\!\!\twoheadleftarrow& \appT{\mAppn{n}}{[\Lookup\trans{M'}{}{\vec x}, \Lookup\trans{\fix M'}{}{\vec x}, \Lookup\etrans{\sigma(x_1)}{}{},\dots,\Lookup\etrans{\sigma(x_n)}{}{}]}\\
&=&\appT{\trans{M'\cdot(\fix M')}{}{x_1,\dots,x_n}}{[\Lookup\etrans{\sigma(x_1)}{}{},\dots,\Lookup\etrans{\sigma(x_n)}{}{}]}\\
&=&\etrans{\sigma, M'\cdot(\fix M'}{}{}.
\ear
\]
\end{description}
This concludes the proof. 
\end{proof}

\begin{proof}[Proof of Theorem~\ref{thm:delayed_equivalence}(\ref{thm:delayed_equivalence1})]
For an \EPCF{} term $M$, an \EPCF{} value $V$ and an explicit substitution $\sigma$, we show
$
	\substseq{\sigma}{ M} \reddd V \Rightarrow (\sigma, M)^* \redd ([\,], V)^*
$
by induction on a derivation of $\substseq{\sigma}{ M} \reddd V$.
\begin{description}
\item[Case $\valrule$:] In this case $M = V = \num n$ for some $n\ge 0$. 
By definition, $(\sigma, \num n)^* = \num n$, so we apply \PCF's rule $\PCFvalrule$ and get $\num n  \redd \num n$.\medskip
\item[Case $\funrule$:] We have $M = \lam x.\esubst{M'}{\sigma'}$ and $V = \lam x.\esubst{M'}{\sigma + \sigma'}$. \\
As $(\sigma,\lam x.\esubst{M'}{\sigma'})^* = ([\,], \lam x.\esubst{M'}{\sigma + \sigma'})^* = \lam x.(\sigma + \sigma', M')^*$ and the latter is a \PCF{} value,  we can apply \PCF's rule $\PCFvalrule$ to conclude 
$
	(\sigma,\lam x.\esubst{M'}{\sigma'})^* \redd ([\,], \lam x.\esubst{M'}{\sigma + \sigma'})^*.
$

\item[Case $\varrule$:] In this case $M = x$, $\sigma(x) = (\rho,N)$ and $\substseq{\rho} N \reddd V$.
By IH $(\rho,N)^*\redd ([],V)^*$. We conclude since $(\sigma,M)^* = (\rho,N)^*$.\medskip

\item[Case $\betarule$:] In this case $M = M_1\cdot M_2$ with $\substseq{\sigma}M_1 \reddd \lam x.\esubst{M_1'}{\rho}$, wlog $x\notin\dom(\rho+\sigma)$, and $\substseq{\rho + [x \leftarrow (\sigma, M_2) ]}{M_1'}\reddd  V$. By IH we obtain $(\sigma, M_1)^* \redd ([\,], \lam x.\esubst{M_1'}{\rho})^*$ and $(\rho + [x \leftarrow (\sigma, M_2)], M_1')^* \redd ([\,], V)^*$. 
By Lemma~\ref{app:lem:substbizarre}\eqref{app:lem:substbizarre1} $(\rho + [x \leftarrow (\sigma, M_2)], M_1')^* = (\rho,M_1')^*\subst{x}{(\sigma, M_2)^*}$. 
By definition, $([\,], \lam x.\esubst{M_1'}{\rho})^* = \lam x.(\rho, M_1')^* $ and $ (\sigma, M_1)^*\cdot(\sigma, M_2)^*= (\sigma, M_1 \cdot M_2)^*$.  Thus 
\[
	\infer[\betarule]{(\sigma, M_1\cdot M_2)^*\redd ([\,], V)^*}{(\sigma, M_1)^* \redd \lam x.(\sigma', M_1')^* & (\sigma',M_1')^*\subst{x}{(\sigma, M_2)^*} \redd ([\,], V)^*}
\]
\item[Case $\fixrule$:] In this case $M = \fix {N}$ and $\substseq{\sigma}{N \cdot (\fix N)} \reddd V$. 
From the IH we get $(\sigma, N \cdot (\fix N))^* \redd ([\,], V)^*$. 
By definition, we have 
$
	(\sigma, N \cdot (\fix N))^* = (\sigma, N)^* \cdot (\fix (\sigma, N)^*)
$ and $\fix (\sigma, N)^* = (\sigma, \fix N)^*$. By applying \PCF's rule $\fixrule$, we obtain
\[
	\infer[\fixrule]{(\sigma,\fix N)^* \redd ([\,], V)^*}{(\sigma, N \cdot (\fix N))^* \redd ([\,], V)^*}
\]
\end{description}
All other cases derive straightforwardly from the IH.
\end{proof}

\begin{proof}[Proof of Theorem~\ref{thm:delayed_equivalence}(\ref{thm:delayed_equivalence2})]
For a \PCF{} program $P$ and \PCF{} value $U$, we show that $P \redd U$ entails:
\[
	\forall (\sigma, M) \in P^\dagger,\, \exists V \in \Val\,.\,(\ \substseq{\sigma}{M}\reddd V \textrm{ and } ([\,], V)\in U^\dagger \ )
\]
By induction on the lexicographically ordered pairs, whose first component is the length of a derivation of $P \redd U$ and second component is $\size{(\sigma,M)}$.

First, consider the case $M = y$ and $\sigma(y) = (\rho,N)$, with $(\rho,N)\in P^\dagger$. 
In this case, the length of the derivation $(\rho,N)^*\redd U$ remained unchanged, while $\size{(\rho,N)} < \size{(\sigma,M)}$. 
Thus, we may use the IH and conclude by applying $\varrule$. Therefore, in the following we assume that $M$ is not a variable.

\begin{description}
\item[Case $\PCFvalrule$:] In this case $P = U$. Given $(\sigma,M)\in P^\dagger$, we distinguish several cases:
\begin{itemize}
\item Case $P = U = \lam x.P_0$. Then, $M = \lam x.\esubst{M_0}{\rho}$ with $P_0 = (\sigma+\rho,M_0)^*$. Setting $V = \lam x.\esubst{M_0}{\sigma+\rho}$ we obtain $M\reddd V$ by $\funrule$ with 
$
	([],V)^* = \lam x.(\sigma+\rho,M_0) = \lam x.P_0 = U.
$
\item
	$P = U = \mathbf{0}$. It follows $M = V = \mathbf{0}$ and $\substseq{\sigma}{\mathbf{0}}\reddd\mathbf{0}$ by $\valrule$.

\item $P = U = \succ(\num n)$, for some $ n \in \nat$, and $M$ is not a variable. 

	There are two possibilities: 
	\begin{itemize}\item $M = \succ(\num n)$ in which case we are done, since $\substseq{\sigma}{\num{n+1}}\reddd\num{n+1}$.
	\item $M = \succ(\num y)$ with $\sigma(y) = (\rho,N)$ and $(\rho,N)\in \num n^\dagger$.
	Again, the length of the derivation $(\rho,N)^*\redd \num n$ is unchanged, while $\size{(\rho,N)} < \size{(\sigma,M)}$. 
	Once applied the IH, we conclude by $\varrule + \succrule$.
	\end{itemize}
\end{itemize}\medskip
\item[Case $\betarule$:] $P = P_1\cdot P_2$ with $P_1\redd \lam x.Q_1$ and $Q_1\subst{x}{P_2}\redd U$ for some $Q_1$. 
Notice that, since $P$ is closed so are $P_1,P_2$ and hence $\FV{Q_1}\subseteq\set{x}$.
Now, $(\sigma,M)^* = P$ entails $M = M_1\cdot M_2$ with $(\sigma,M_1)\in P_1^\dagger$ and $(\sigma,M_2)\in P_2^\dagger$.
By ind.\ hyp., there is $V'\in\Val$ such that $\substseq{\sigma}{M_1}\reddd V_1$ with $([],V_1)\in (\lam x.Q_1)^\dagger$.
This implies $V_1 = \lam x.\esubst{N_1}{\rho}$ for some $(\rho,N_1)\in Q_1^\dagger$.
By Lemma~\ref{app:lem:substbizarre}\eqref{app:lem:substbizarre2}, we get $(\rho+[x\leftarrow(\sigma,M_2)],N_1)\in (Q_1\subst{x}{P_2})^\dagger $. 
By ind.\ hyp., there is $V\in\Val$ such that $\substseq{\rho+[x\leftarrow(\sigma,M_2)]}{N_1}\reddd V$ and $V\in U^\dagger$. Conclude by \EPCF's $\betarule$.\medskip

\item[Case $\fixrule$:] $P = \fix Q$ and $Q\cdot(\fix Q)\redd V$. Then $(\sigma,M)\in P^\dagger$ entails $M = \fix N$ with $(\sigma,N)\in Q^\dagger$. It follows that $(\sigma,N\cdot(\fix N))\in(Q\cdot(\fix Q))^\dagger$, therefore by IH we get $\substseq{\sigma}{N\cdot(\fix{N})}\reddd V$ for some $V\in U^\dagger$.
We conclude by applying \EPCF's rule $\fixrule$.
\end{description}
All other cases derive straightforwardly from the IH.
\end{proof}

\end{document}